\documentclass[a4paper,UKenglish,cleveref, autoref, thm-restate, final]{lipics-v2021}
\newif\iffinal

\iffinal
\else
\nolinenumbers 
\fi

\newif\ifnoproofs

\ifnoproofs
  \usepackage{environ}
  \NewEnviron{killcontents}{}

\else
\fi

\iffinal
\else
  \IfFileExists{ayrat_setup}{
    \paperwidth 15.8cm
    \paperheight 24cm
    \oddsidemargin  0.8cm
    \evensidemargin 0.8cm
    \topmargin      -1.3cm
  }{}
\fi

\usepackage{standalone}
\usepackage{amsfonts}
\usepackage{amsmath}
\usepackage{mathtools}
\usepackage{amssymb}
\usepackage{bbm}
\usepackage{paper-includes/temporal}
\usepackage{graphicx}
\usepackage{caption}
\usepackage{xspace}
\usepackage{subcaption}
\usepackage{stmaryrd}  
\usepackage{doi}       
\usepackage[cal=esstix]{mathalfa}
  
\usepackage{xparse}    
\usepackage{diagbox}   
\usepackage{booktabs}  
\usepackage{array}     
\usepackage{multirow}  
\usepackage{cases}     
\usepackage[T1]{fontenc}

\usepackage[x11names,svgnames]{xcolor}  

\usepackage{float} 

\usepackage{tikz}
  \usetikzlibrary{arrows,automata,positioning,matrix,calc,petri,backgrounds,shapes,decorations.pathreplacing}

  \tikzset{%
    initial text={}%
}
  \tikzset{%
    input state/.style={state,fill=\inColor,fill opacity=0.5,text opacity=1,rectangle,minimum size=7mm,inner sep=1pt},%
   output state/.style={state,fill=\outColor,fill opacity=0.5,text opacity=1,circle,minimum size=8.25mm,inner sep=1pt},%
   trans state/.style={state,fill=lightgray,fill opacity=0.5,text opacity=1,circle,minimum size=8.25mm,inner sep=1pt},%
  }

\usepackage{multicol}

\definecolor{rose}{rgb}{0.75, 0, 0.2}
\definecolor{myblue}{rgb}{0, 0.2, 0.7}

\iffinal
  \newcommand{\leo}[1]{}
  \NewDocumentCommand \todo{g} {}
\else

  \NewDocumentCommand \todo{g} { \IfValueTF{#1} {{\color{Red}{\scriptsize /#1/}}} {{\color{Red}{\small \tt todo}}} }
  \newcommand{\leo}[1]{}

\fi

\usepackage{pifont}
\newcommand{\cmark}{\ding{51}}%
\newcommand{\xmark}{\ding{55}}%

\usepackage{hyperref}   
\hypersetup{
    colorlinks,
    linkcolor=black,
    citecolor=black,
    urlcolor=black,
}

\theoremstyle{definition}
\newtheorem*{definition*}{Definition}
\newtheorem{fact}{Fact}
\newtheorem*{fact*}{Fact}
\theoremstyle{remark}
\newtheorem*{example*}{Example}
\newtheorem*{discussion*}{Discussion}

\newcommand{\bbN}{\mathbb{N}}

\newcommand{\bbQ}{\mathbb{Q}}
\newcommand{\bbQP}{\mathbb{Q}_+}
\newcommand{\bbZ}{\mathbb{Z}}

\newcommand{\Impl}{\Rightarrow}

\newcommand{\Implied}{\Leftarrow}
\newcommand\LAND\bigwedge
\newcommand\LOR\bigvee

\newcommand\x{\mkern 1mu{\times}\mkern 1mu}

\newcommand\li{\begin{itemize}}
\newcommand\il{\end{itemize}}
\renewcommand{\-}{\item}  
\newcommand\lo{\begin{enumerate}}
\newcommand\ol{\end{enumerate}}

\usepackage{calc}

\newcommand\trans[1][]{
  \xrightarrow[{\raisebox{1.25ex-\heightof{$\scriptstyle#1$}}[0pt]{$\scriptstyle#1$}}]%
}

\let\emptyset\varnothing

\newcommand\G\always\xspace

\renewcommand{\:}{\colon}

\let\oldsharp\#
\renewcommand{\#}{\raisebox{0.1em}{\scalebox{0.8}{{\oldsharp}}}}   

\newcommand{\myparagraph}[1]{\subparagraph{#1.}}

\definecolor{ggray}{rgb}{0.9, 0.9, 0.9}

\newcommand{\paritraw}[1]{\smallskip\noindent {\it #1}}
\newcommand{\parit}[1]{\paritraw{#1.}}

\newcommand{\size}[1]{\left| #1 \right|}
\newcommand{\length}[1]{\left\lvert #1 \right\rvert}

\usepackage{mathtools}

\newcommand{\specialcellC}[2][c]{%
  \begin{tabular}[#1]{@{}c@{}}#2\end{tabular}%
}

\newcommand{\da}{\downarrow\!}

\newcommand{\clen}{\mathit{cl\hspace{-0.33mm}en}}
\newcommand{\clenR}[2]{\pi_{#1,#2}}
\newcommand{\inWord}{\mathsf{x}}
\newcommand{\lab}{\mathsf{lab}}

\renewcommand\d{{\mathcal d}}
\newcommand\dI{\d^{\frak i}}
\newcommand\dO{\d^{\frak o}}
\renewcommand{\v}{\nu}

\newcommand{\Ss}{S_\textit{synt}}
\newcommand{\Ts}{T_\textit{synt}}

\renewcommand{\a}{\seq{a}}
\newcommand{\ak}{\seq{a}_k}
\newcommand{\aS}{\seq{a}_S}

\newcommand{\Wf}{W^{\textsf{F}}_{S,k}}
\newcommand{\Wqf}{W^{\textsf{QF}}_{S,k}}

\newcommand\wB{\ensuremath{\color{black}\omega}B\xspace}
\newcommand\wS{\ensuremath{\color{black}\omega}S\xspace}

\newcommand\D{\mathcal D}
\newcommand\bbD{\mathbb D}
\newcommand\DN{(\bbN,<,0)}
\newcommand\DZ{(\bbZ,<,0)}
\newcommand\DQ{(\bbQ,<,0)}
\newcommand\DQP{(\bbQP,<,0)}

\newcommand{\tst}{\textnormal{\textsf{tst}}}
\newcommand{\asgn}{\textnormal{\textsf{asgn}}}
\newcommand{\Tst}{\textnormal{\textsf{Tst}}}
\newcommand{\MTst}{\textnormal{\textsf{MTst}}}
\newcommand{\Asgn}{\textnormal{\textsf{Asgn}}}

\newcommand\tstI{\tst^\frak i}
\newcommand\tstSI{\tst^{{\scalebox{0.66}{\!$S$}}\frak i}}
\newcommand\tstO{\tst^\frak o}
\newcommand\tstSO{\tst^{{\scalebox{0.66}{\!$S$}}\frak o}}

\newcommand\asgnI{\asgn^\frak i}
\newcommand\asgnSI{\asgn^{{\scalebox{0.66}{\!$S$}}\frak i}}
\newcommand\asgnO{\asgn^\frak o}
\newcommand\asgnSO{\asgn^{{\scalebox{0.66}{\!$S$}}\frak i}}

\newcommand{\AW}{\textsf{AW}}
\newcommand{\TW}{\textsf{TW}}
\newcommand{\F}{\textsf{FEAS}}
\newcommand{\QF}{\textsf{QFEAS}}
\newcommand\lasso{\mathit{lasso}}

\newcommand{\seq}[1]{\mkern 2mu\overline{\mkern-2mu#1}}

\newcommand\comp{\raisebox{0.1mm}{\scalebox{0.87}{$\kern 0.1em{\parallel}$}}}

\newcommand*{\inColor}{Red4}
\newcommand*{\outColor}{SeaGreen4}
\newcommand*{\inColored}[1]{\textcolor{\inColor}{#1}}
\newcommand*{\outColored}[1]{\textcolor{\outColor}{#1}}
\newcommand*{\store}[1]{\downarrow{}\hspace{-0.15em} #1}
\newcommand*{\indata}{\star}

\newcommand{\tr}{\triangleleft}

\makeatletter
\newcommand*{\stx}[1]{\mathpalette\giusti@stx{#1}}
\newcommand*{\giusti@stx}[2]{%
  \mbox{%
    \medmuskip=\thinmuskip
    \thickmuskip=\thinmuskip
    $\m@th#1#2$%
  }%
}
\makeatother
\newcommand*{\valstore}[3]{#1\stx{[#2 \leftarrow #3]}}

\title{A Generic Solution to Register-bounded Synthesis\texorpdfstring{\\}{ }with an Application to Discrete Orders \texorpdfstring{\scalebox{0.85}{(full version)}}{(full version)}}

\titlerunning{A Generic Solution to Register-bounded Synthesis}

\author{L\'eo Exibard}{Reykjavik University, Iceland}{}{}{}
\author{Emmanuel Filiot}{Universit\'e libre de Bruxelles, Belgium}{}{}{}
\author{Ayrat Khalimov}{Universit\'e libre de Bruxelles, Belgium}{}{}{}

\authorrunning{L.\ Exibard, E.\ Filiot, A.\ Khalimov}  

\Copyright{L\'eo Exibard, Emmanuel Filiot, Ayrat Khalimov} 



\funding{This work was supported by the Fonds de la Recherche
  Scientifique - F.R.S.-FNRS under the MIS project F451019F. Emmanuel
  Filiot is a senior research associate at F.R.S.-FNRS. L\'eo Exibard was supported by the project ‘Mode(l)s of Verification and Monitorability (MoVeMnt)’ (no. 217987-051) of the Icelandic Research Fund.}

\begin{document}

\setlength\abovedisplayskip{4pt}
\setlength\belowdisplayskip{4pt}

\maketitle

\begin{abstract}
  We study synthesis of reactive systems interacting with environments using an infinite data domain.
  A popular formalism for specifying and modelling such systems is register automata and transducers.
  They extend finite-state automata by adding registers to store data values and
  to compare the incoming data values against stored ones.
  Synthesis from nondeterministic or universal register automata is undecidable in general.
  However, its register-bounded variant, where additionally a bound on the number of registers in a sought
  transducer is given, is known to be decidable for universal
  register automata which can compare data for equality, i.e., for data domain $(\bbN,=)$.
  This paper extends the decidability border to the domain $(\bbN,<)$ of natural numbers with linear order.
  Our solution is generic:
  we define a sufficient condition on data domains (regular approximability)
  for decidability of register-bounded synthesis.
  The condition is satisfied by natural data domains like $(\bbN,<)$.
  It allows one to use simple language-theoretic arguments and avoid technical game-theoretic reasoning.
  Further, by defining a generic notion of reducibility between data domains,
  we show the decidability of synthesis in the domain
  $(\bbN^d,<^d)$ of tuples of numbers equipped with
  the component-wise partial order
  and in the domain $(\Sigma^*,\prec)$ of finite strings with the prefix relation.
\end{abstract}

\leo{I suggest adding theorem/definition names in the citations. E.g., when we mention oligomorphic domains, write \cite[Definition~3.9]{Bojanczyk19}.}
\leo{I would leave constants in the signature as we actually need them (in particular 0 in $\bbN$). The change I was suggesting was rather to encode them as predicates, e.g. $0(x)$ holds if and only if $x=0$. Otherwise, we get problems when we have expressions like $p(x_1,...,c,...,x_n)$, because $c$ is stored nowhere.}
\section{Introduction}
\myparagraph{Synthesis}
Reactive synthesis aims at the automatic construction of an interactive system from its specification.
A system is usually modelled as a transducer.
In each step, it reads an input from the environment and produces an output.
In this way, the transducer, reading an infinite sequence of inputs, produces an infinite sequence of outputs.
Specifications are modelled as a language of desirable input-output sequences.
The synthesis problem then asks to automatically construct a transducer whose input-output sequences belong to a given specification.
Traditionally~\cite{PR89a,DBLP:reference/mc/BloemCJ18}, the inputs and outputs have been modelled as letters from a finite alphabet.
This, however, limits the application of synthesis.
Recently researchers have started investigating
synthesis of systems working on data domains~\cite{ESK14,KMB18,EFR21,KK19,DBLP:conf/fossacs/BerardBLS20,EFK21}.

\myparagraph{Automata as specifications}
In the finite-alphabet setting,
specifications are usually given as logical formulas
and a synthesiser performs a series of translations: first, from the formula to an automaton, then from the automaton to a game, and
finally it searches for a winning strategy in the game.
It is the second step, from automata to games, that captures the game-theoretic essence of synthesis,
whereas the first step is an orthogonal problem of finding a convenient logical formalism.
In the context of synthesis over data domains, this first step is problematic as there is no decidable,
and expressive enough, logic having a corresponding automaton model.
For that reason, in this paper we focus on the second step and use automata for specifications.

\myparagraph{Register automata}
A well-studied automata formalism for specifying and modelling data systems are
register automata and transducers~\cite{KF94,NSV01,KZ08,Tze11}.
Register automata extend classical finite-state automata to infinite alphabets
$\bbD$ by introducing a finite number of \emph{registers}. In each
step, the automaton reads a data value $\d\in\bbD$, compares it with the values held in
its registers, then depending on this comparison it decides to store
$\d$ into some of its registers, and finally moves to a successor state. This way, it
builds a sequence of \emph{configurations} (pairs of state and register values)
representing its run on reading a word from $\bbD^\omega$: it is accepting if the
visited states satisfy a certain condition, e.g.\ parity. Transducers are
similar except that in each step they also output the content of one register.

\myparagraph{Universal register automata}
Unlike classical finite-state automata,
the expressive power of register automata depends on whether they are
deterministic, nondeterministic, or universal (a.k.a.\ co-nondeterministic).
Among these, universal register
automata suit synthesis best. First, they can specify request-grant properties:
every requested data shall be eventually outputted.
This is a key property in reactive synthesis, and
in the data setting it can be expressed by a universal register automaton
but not by a nondeterministic one. Furthermore, universal register
automata are closed, in linear time, under intersection.
Hence they allow for succinct conjunction of properties,
which is desirable in synthesis as specifications usually consist of many independent properties.
Finally, in the register-free setting
universal automata are often used to obtain synthesis methods
feasible in practice~\cite{KV05c,SF07,FJR09,DBLP:reference/mc/BloemCJ18}.

\myparagraph{Data domains with order}
Another factor affecting expressivity of register automata
is the data-comparison operators.
Originally, register automata compared data for equality only,
i.e., operated in data domain $(\bbD,=)$~\cite{KF94}.
This limits synthesis applications
as we cannot specify priority arbiters~\cite{BKY00} that should give a resource to a requesting process with the lowest ID.
Such properties require data domains with linear order $<$ (in addition to $=$).
Further, there are data domains with dense order, like $(\bbQ,<)$, and
those with discrete order, like $(\bbN,<)$.
The domain $(\bbQ,<)$ is well-suited for abstracting physical phenomena like changing temperature in a room.
However, for abstracting hardware,
the domain $(\bbN,<)$ suits better as it excludes Zeno-like behaviours
(when a process ID gets infinitely closer to another ID but never reaches it).
The domain $(\bbN,<)$ is also interesting from the theoretical point of view
as it demands new proof techniques.

\myparagraph{Known synthesis results for register automata}
Already for $(\bbD,=)$,
the synthesis problem of register transducers from universal register automata is undecidable~\cite{ESK14,EFR21}.
Decidability is recovered in the deterministic case~\cite{EFR21,EFK21},
but, as argued above, universal automata are more desirable in synthesis.
To circumvent undecidability,
the works~\cite{KMB18,EFR21,KK19} studied \emph{register-bounded} synthesis:
given a universal register automaton \emph{and a bound $k$} on the number of transducer registers,
return a $k$-register transducer realising the automaton or `No' if no such transducer exists.
They showed the decidability of register-bounded for $(\bbD, =)$,
and it is not hard to adapt their techniques to $(\bbQ,<)$ and other oligomorphic domains~\cite{Bojanczyk19},
however the domain $(\bbN,<)$ remained elusive.
Tables \ref{tbl:reg-unbound-results} and \ref{tbl:reg-bound-results} summarise known and new results,
where DRA/NRA/URA stand for deterministic/nondeterministic/universal register automata.

\begin{table}
  \newcommand\DEC{{\color{Green}\cmark}}
  \newcommand\UNDEC{{\color{red}\xmark}}
  \centering
  \begin{minipage}[c]{0.47\textwidth}
    \centering
    \begin{tabular}{c|ccc}
      & $(\bbD,=)$    & $(\bbQ,<)$            & $(\bbN,<)$      \\
      \hline
      DRA    & \DEC\cite{EFK21}  & \DEC\cite{EFK21}   & \DEC\cite{EFK21}    \\
      NRA    & \UNDEC\cite{EFR21}  & \UNDEC   & \UNDEC    \\
      URA    & \UNDEC\cite{EFR21}  & \UNDEC   & \UNDEC
    \end{tabular}%
  \caption{Decidability of register-unconstrainted synthesis.}\label{tbl:reg-unbound-results}
  \end{minipage}
  ~~~~
  \begin{minipage}[c]{0.47\textwidth}
    \centering
    \begin{tabular}{c|ccc}
      & $(\bbD,=)$    & $(\bbQ,<)$            & $(\bbN,<)$      \\
      \hline
      DRA    & \DEC\cite{KMB18}  & \DEC\cite{ExibardThesis}   & \DEC [\sc this paper]    \\
      NRA    & \UNDEC\cite{EFR21}  & \UNDEC   & \UNDEC    \\
      URA    & \DEC\cite{KMB18}  & \DEC\cite{ExibardThesis}   & \DEC [\sc this paper]
    \end{tabular}%
  \caption{Decidability of register-bounded synthesis.}\label{tbl:reg-bound-results}
  \end{minipage}
\end{table}

\myparagraph{Contributions}
We prove that register-bounded synthesis is decidable for $(\bbN,<)$
in time doubly exponential in the number of registers of the specification automaton and of the sought transducer.
Our procedure is effective: it constructs a transducer if one exists.
When the total number of registers is fixed,
it is {\sc ExpTime-c}, matching the complexity of classical (register-free) synthesis.
This result generalises the works of~\cite{EFR21,KK19,KMB18} on $(\bbD,=)$.
We then extend the decidability boundary farther to include the domain $(\bbN^d,<^d)$ of tuples of naturals with the component-wise partial order, and the domain $(\Sigma^*,\prec)$ of strings with the prefix relation.

\myparagraph{Technical contributions}
Our proof technique is generic and greatly simplifies the task of proving new synthesis decidability results by removing the need to reason about synthesis alltogether. We now describe the technique in detail.

The key idea of existing approaches~\cite{KMB18,EFR21,KK19}
is to reduce the register-bounded synthesis problem in a data domain to a two-player Church game with a finite alphabet and an $\omega$-regular winning condition.
In such a game, two players alternately play for an infinite number of rounds.
Adam, modelling the environment, picks a test over the $k$ registers describing how its input data compares with the current content of the registers of a sought transducer.
Eve, modelling the system,
picks a subset of the $k$ registers, meant to store the data, and a register whose value is meant for output.
No data are manipulated in the game.
Infinite plays in the game induce infinite sequences of tests, assignments, and outputs over the $k$ registers,
called \emph{action words}; they are over a \emph{finite} alphabet.
Action words are meant to abstract data words;
an action word is \emph{feasible} if there is at least one data word that satisfies all its tests and assignments.
The reduction ensures that any strategy of Eve winning in the game
can be converted into a $k$-register transducer realising the specification, and vice versa.
To this end,
the game winning condition declares a play to be won by Eve
if
all data words satisfying the action word induced by the play are accepted by the specification automaton.
In particular,
a play whose action word is \emph{un}feasible is won by Eve
as it does not correspond to any environment-system interaction in the data domain.
In the case of $(\bbD,=)$,
such winning conditions are known to be $\omega$-regular~\cite{KMB18,EFR21,KK19}.
However,
in $(\bbN,<)$ the set of feasible action words is not $\omega$-regular~\cite{EFK21},
and neither is the winning condition.
Such winning conditions could be expressed by nondeterministic \wS automata~\cite{BC06},
but games with such objectives are not known to be decidable, to the best of our knowledge.

To overcome the latter obstacle,
we introduce the notion of $\omega$-regularly approximable
(\emph{regapprox}) data domains.
A regapprox data domain has an $\omega$-regular over-approximation of the
set of feasible action words that is exact on the lasso-shaped action
words (of the form $uv^\omega$). Thus, in regapprox domains the
set of feasible lasso-shaped action words is $\omega$-regular.
This allows us to avoid dealing with non-$\omega$-regularity and
reduce synthesis to solving classic $\omega$-regular games.
Our first technical contribution is the generic decidability result:
\begin{quote}%
  \centering
  \emph{For regapprox domains, register-bounded synthesis from URA is decidable}.
\end{quote}
The procedure is constructive: for realisable specifications it outputs a transducer.
Note that all oligomorphic domains~\cite{Bojanczyk19}, e.g.\ $(\bbD,=)$ and $(\bbQ,<)$,
are regapprox,
because their sets of feasible action words are $\omega$-regular,
so our result subsumes works~\cite{EFR21,KK19,KMB18}.
For $(\bbN,<)$,
we construct its over-approximation relying on the result~\cite{EFK21},
and then instantiate the theorem.

There are many domains with discrete order resembling $(\bbN,<)$:
the domain $(\bbZ,<)$ of integers,
the domain $(\bbN^d,<^d)$ of tuples of naturals with the component-wise partial order, and
even the domain $(\Sigma^*,\prec)$ of strings with the prefix relation.
To further simplify decidability proofs on these domains,
we define a natural and generic notion of reducibility between data domains.
Intuitively, a data domain $\bbD$ \emph{reduces} to $\bbD'$ if there is a rational transduction
that relates action words in $\bbD$ and $\bbD'$ while preserving feasibility.
Our second technical contribution is the reduction result:
\begin{quote}%
  \centering
  \emph{If $\bbD$ reduces to $\bbD'$, and $\bbD'$ is regapprox,
  then $\bbD$ is regapprox}.
\end{quote}
This implies that a synthesis procedure for $\bbD'$ can be used to solve synthesis in $\bbD$.
We illustrate the technique by reducing to $(\bbN,<)$
the domains $(\bbN^d,<^d)$ and $(\Sigma^*,\prec)$.
The reduction for $(\Sigma^*,\prec)$ relies on the work~\cite{DD16}.
These reductions entail the decidability of register-bounded synthesis on these domains.

\myparagraph{Related works}
We already mentioned the works~\cite{KMB18,EFR21,KK19,ExibardThesis}
on synthesis of register transducers in domains $(\bbD,=)$ and $(\bbQ,<)$,
and that our result generalises them for the case of URAs.
The paper~\cite{EFK21} studies \emph{Church's synthesis} for DRA specifications,
where a \emph{data} strategy not necessarily with finitely-many states is sought.
However, they show that considering register transducers is sufficient,
with with the number of registers equal that of the specification automaton.
Hence our register-bounded synthesis procedure for URAs can also be used to solve the Church's synthesis problem.

Another formalism for specifications of data systems is that of \emph{variable automata}~\cite{GKS10}.
The paper~\cite{FK20} studies synthesis of symbolic transducers
from specifications given in a fragment of nondeterministic variable automata.
They solve synthesis for data domain $(\bbQ,<)$ and leave the domain $(\bbN,<)$ for future work.
Variable automata are incomparable with register automata,
and their particular fragment cannot express request-grant properties of arbiters
that we believe is desirable in synthesis.

Our proof techniques resemble those from some works on satisfiability of data logics.
Constraint LTL~\cite{DD07} extends Linear Temporal Logic (LTL) by atoms
allowing one to compare data values within the horizon or pre-defined length.
The satisfiability of this logic is decidable for
data domains $(\bbD,=)$, $(\bbQ,<)$, $(\bbN,<)$~\cite{DD07}, and $(\Sigma^*,\prec)$~\cite{DD16}.
Their proof technique relies on the abstraction of data values at different moments
by relations between each other.
For the data domain $(\bbN,<)$,
they additionally prove that considering lasso-shaped witnesses of satisfiability is sufficient.
Our generic synthesis result uses a similar idea by defining regapprox domains.
We note that formulas in Constraint LTL can always be translated into universal register automata
(which are more expressive)~\cite{ST11}.
Hence our approach can be used to solve register-bounded synthesis from Constraint LTL.

The papers~\cite{FKPS19,MB21} suggest a sound/incomplete procedure to
synthesis from Temporal Stream Logic.
This logic extends LTL by adding the atoms
that are either first-order predicate terms or are assignments of variables to a first-order function term.
Similarly, transducers can test data using the predicate terms and update its values by the function terms.
A transducer satisfies a specification if it does so
under \emph{every} interpretation of predicates and functions.
It is possible to model domains like $(\bbD,=)$ and $(\bbQ,<)$ in their formalism,
by encoding the axioms for $>$ and $=$ into specification.
This would give a sound/\emph{in}complete synthesis approach.
Our approach is less general but retains the completeness.

More generally, our notion of regular approximation echoes a general
idea common to verification techniques, for example of programs manipulating data
variables (see, e.g.,~\cite{JPR18}), to abstract concrete behaviours by regular ones. When an
over-approximation is used, it is guaranteed that if the abstract
program satisfies some safety properties, so does the concrete
program. This yields sound algorithm which are not necessarily
complete. Here in the context of register automata, instead, we require that the over-approximation is
exact on lasso-like executions, and show that this implies completeness (for
the synthesis problem).


\section{Synthesis Problem}\label{sec:problem}

Let $\bbN = \{0,1,\dots\}$ denote the set of natural numbers including $0$.

\myparagraph{Data domain and data words} A \emph{data domain} is a tuple $\D = (\bbD, P, C, c_0)$
consisting of an infinite countable set $\bbD$ of \emph{data values},
a finite set $P$ of interpreted \emph{predicates} (predicate names with arities and their interpretations)
which must contain the equality predicate $=$,
a finite set $C \subset \bbD$ of \emph{constants}, and a distinguished
\emph{initialiser} constant $c_0 \in C$.  For example,
$(\bbN,\{<,=\},\{0\},0)$ is the data domain of natural numbers with
the usual interpretation of $<$, $=$, and $0$. In the tuple notation, we often
omit the brackets, as well as the mention of $=$ and of $c_0$ when the initialiser
constant is clear from the context.
E.g., we write $\DN$ for
$(\bbN,\{<,=\},\linebreak\{0\},0)$. Another familiar example is $(\bbZ,<,0)$,
which is the data domain of integers with the usual $<$, $=$, and
$0$. Throughout the paper we assume that the satisfiability
problem of quantifier-free formulas built on the signature $(P,C)$ is
decidable in $\D$, and whenever we state complexity results, the satisfiability problem is additionally assumed to be
decidable in \textsc{PSpace}. This is the case for all data domains considered in this paper.
Finally, \emph{data words} are infinite sequences $\d_0 \d_1 \ldots \in \bbD^\omega$,
and for two sets $I$ and $O$ and a language $L \subseteq (I\cdot O)^\omega$,
we call $I$ and $O$ its \emph{input} and \emph{output} alphabets respectively.\leo{I replaced $\Sigma,\Gamma$ with $I,O$ to highlight that they might not necessarily be finite.} 

\myparagraph{Action words}
Fix a data domain $\D = (\bbD, P, C, c_0)$ and a finite set $R$ of elements called \emph{registers}.
A register \emph{valuation} (over $\D$) is a mapping $\v : R \to \bbD$.
Given a valuation $\v$, a variable $x$ (not necessarily in $R$), and a data value $\d \in \bbD$,
define $\valstore{\v}{x}{\d}$ to be the valuation $R \cup \{x\} \to \bbD$ that maps $x$ to $\d$ and every $r \in R\setminus\{x\}$ to $\v(r)$.
We extend this notation to any finite set $A = \{a_1,\dots,a_n\}$ by letting $\valstore{\v}{A}{\d} = \v\stx{[a_1 \leftarrow \d]} \dots \stx{[a_n \leftarrow \d]}$.

A \emph{test} (over $\D$) is a conjunction $(\wedge)$ of distinct literals over predicates $P$ and constants $C$, encoded as a set of literals $p(x_1,\dots,x_a)$ and $\neg p(x_1,\dots,x_a)$, where $p \in P$, $a$ is the arity of $p$ and $x_1,\dots,x_a \in R \cup C \cup \{\indata\}$.
The symbol $\indata$ is a fresh symbol used as a placeholder for incoming data values.
By convention, $\wedge \varnothing = \top$, and the empty set encodes the test that is always true. Depending on the context, we use the formula or set notation.
A register valuation $\v : R \to \bbD$ and data value $\d \in \bbD$ \emph{satisfy} a test $\varphi$, written $\v,\d \models \varphi$, if $\valstore{\v}{\indata}{\d}$ satisfies $\varphi$, where predicates and constants are interpreted in the data domain $\D$.
A test $\varphi$ is \emph{maximal} if it specifies the relation between all variables and constants wrt.\ the predicates, i.e.\ it is a maximally consistent conjunction of literals: $\varphi = \bigwedge_{\substack{p \in P \\ p \text{ of arity }r}} \bigwedge_{\substack{x_1,\dots,x_r \\ \in R \cup C \cup \{\indata\}}} l_{p,x_1,\dots,x_r}$, where $l_{p,x_1,\dots,x_r} \in \{p(x_1,\dots,x_r), \neg p(x_1,\dots,x_r)\}$. Maximal tests are mutually exclusive: a given valuation cannot satisfy simultaneously two of them. Observe that a test is equivalent to a (possibly exponential) disjunction of maximal ones.
Let $\Tst_R^\D$ denote the set of all possible tests over registers $R$ in domain $\D$, and $\MTst_R^\D \subset \Tst_R^\D$ the subset of maximal ones.
\begin{example*}
  Consider domain $\DN$ and $R = \{r\}$.
  Atomic formulas are $r < \indata$, $\indata = r$, $r < 0$, $\indata = 0$, etc. The test $0<r \wedge r <\indata$ specifies that the content of register $r$ is strictly positive and that the incoming data is greater than it. It is not maximal, since it does not contain the atoms $0 < \indata$, $\neg (\indata = r)$, $\neg (\indata = 0)$, $\neg (r = 0)$. For readability, we write $0 < r < \indata$. 
\end{example*}


An \emph{assignment} is a set $\asgn \subseteq R$ of registers
meant to store the current input data value.
Let $\Asgn_R = 2^R$ denote the set of all possible assignments.
An \emph{action} is a pair $(\tst,\asgn) \in \Tst_R\x\Asgn_R$.
We now describe how valuations are updated:
given a valuation $\v$, a data value $\d$, a test $\tst$ and an assignment $\asgn$,
we say that the valuation $\v'$ is the \emph{successor} of $\v$ following action $(\tst,\asgn)$ on reading $\d$, written $\v \trans{\d,\tst,\asgn} \v'$,
if
the data value satisfies the test, i.e.\ $\v,\d \models \tst$, and
$\v' = \valstore{\v}{\asgn}{\d}$.

An \emph{automaton action word}, or simply \emph{action word}, is an infinite sequence of actions $\a = (\tst_0,\asgn_0)(\tst_1,\asgn_1)\ldots\in(\Tst_R\x\Asgn_R)^\omega$.
It is \emph{feasible} by a sequence of valuation-data pairs $(\v_0, \d_0) (\v_1, \d_1) \ldots$
if $\v_0: r \in R \mapsto c_0$,
i.e.\ $\v_0$ maps every $r\in R$ to $c_0$,
and for all $i$: $\v_{i} \trans{\d_i,\tst_i,\asgn_i} \v_{i+1}$.
We then say that the data word $\d_0 \d_1 \dots$ is \emph{compatible} with $\a$.
Let $\AW_R^\D$ denote the set of action words over $R$ in $\D$, and $\F_R^\D$ the subset of feasible ones. We may write either $\AW_R$, or $\AW^\D$ or just $\AW$ when $\D$, $R$ or both are clear from the context, similarly for $\F$. \label{page:def:feas}

\begin{example*}
Consider domain $\DN$ and $R = \{r\}$.
For $r \in R$, the assignment $\{r\}$ is denoted $\da r$.
The action word $(0<\indata,\da r)(\indata<r,\da r)^\omega$ is unfeasible in $\DN$,
because it requires having an infinite chain of strictly decreasing values, which is not possible since $\bbN$ is well-founded.
The same action word can be interpreted in $\DZ$ and in $\DQ$ and there it is feasible,
as well as in $\DQP$ since $\bbQP$ is dense.
\end{example*}

\myparagraph{Register automata}
A register automaton over data domain $\D$ is a tuple $S = (Q, q_0, R, \delta, \alpha)$, where $Q$ is a finite set of \emph{states} containing the \emph{initial} state $q_0$, $R$ is a finite set of \emph{registers}, $\delta \subseteq Q \x \Tst_R \x \Asgn_R \x Q$ is a \emph{transition relation}, and $\alpha : Q \to \{1,...,c\}$ is a \emph{priority function} where $c$ is the priority \emph{index}. A \emph{configuration} of $S$ is a pair $(p,\v)\in Q\x\bbD^R$; it is \emph{initial} if $p = q_0$ and $\v: r\in R \mapsto c_0$. The configuration $(q,\v')$ is a \emph{successor} of $(p,\v)$ on reading data value $\d \in \bbD$ and taking transition $p' \trans{\tst,\asgn} q' \in \delta$,
written $(p,\v) \trans[S]{\d,\tst,\asgn} (q,\v')$ or simply $(p,\v) \trans[S]{\d} (q,\v')$,
if $p = p'$, $q = q'$ and $\v \trans{\d,\tst,\asgn} \v'$, i.e. $\v,\d \models \tst$ and $\v' = \valstore{\v}{\asgn}{\d}$.

A \emph{run} of $S$ over a data word $\d_0 \d_1 \ldots$ is a sequence of configurations $\rho = (q_0,\v_0) (q_1,\v_1)\ldots$ such that $(q_0,\v_0)$ is initial and for every $i$, $(q_{i+1},\v_{i+1})$ is a successor of $(q_i,\v_i)$ on reading $\d_i$, on taking some transition $q_i \trans{\tst_i,\asgn_i} q_{i+1} \in \delta$. We then say that the automaton action word $(\tst_0,\asgn_0)(\tst_1,\asgn_1)\ldots$ \emph{labels} $\rho$; note that it is feasible by $\v_0\d_0\v_1\d_1\ldots$.
The run $\rho$ is \emph{accepting} if
the maximal priority appearing infinitely often in $\alpha(q_0)\alpha(q_1) \ldots$ is even,
otherwise it is \emph{rejecting}.
A data word may have several runs of $S$.
For \emph{universal} register automata, abbreviated URA,
a word is \emph{accepted} if all its runs are accepting;
for \emph{nondeterministic} automata, there should be at least one accepting run.
The set of all data words over $\D$ accepted by $S$ is called the
\emph{language} of $S$ and denoted $L(S)$. We may write $L_\D(S)$ to
emphasise that $L(S)$ is defined over $\D$.

A \emph{finite (parity) automaton} (without registers) is a tuple $(\Sigma,Q,q_0,\delta,\alpha)$,
where $\Sigma$ is a finite alphabet,
$\delta \subseteq Q \x\Sigma \x Q$,
and the definition of runs, accepted words, and language is standard.
Such automata operate on words from $\Sigma^\omega$.

\parit{Syntactical language of a register automaton} A register automaton $S = (Q,q_0,R,\delta,\alpha)$ can be treated syntactically, it then
induces a universal finite automaton $\Ss = (\Sigma,Q,q_0,\delta,\alpha)$
with $\Sigma = \Tst_R\x\Asgn_R$. Note that since $\Ss$ is universal, words that have no run are accepted.
Notice that
the language of $\Ss$ may contain action words which are not feasible.
\begin{example*}\label{page:example:UpWeGo}
  Consider the automaton of Figure~\ref{fig:UpWeGo}.
  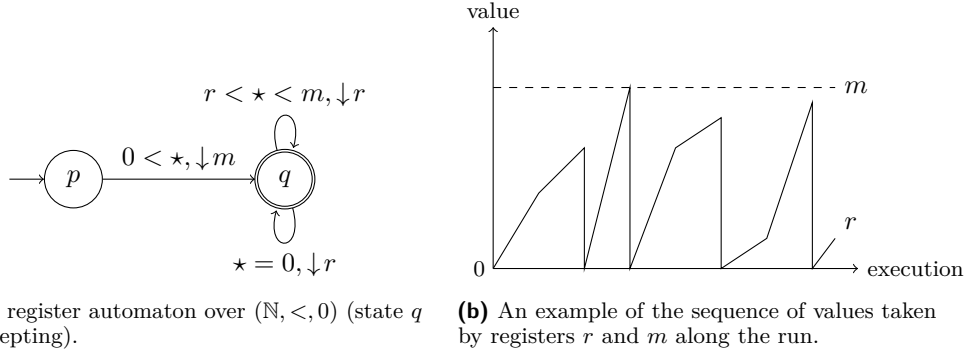
\begin{figure}[tb]
    \centering
    \begin{subfigure}[b]{0.45\textwidth}
      \centering
      \begin{tikzpicture}[->,node distance=2cm,every state/.style={inner sep=1.7mm, minimum size=5mm}]
      \node[state, initial]               (0) {$p$};
      \node[state, right= of 0,accepting] (1) {$q$};

      \path (0) edge             node[above] {$0 < \indata, \store{m}$} (1);
      \path (1) edge[loop above] node[above] {$ r < \indata < m, \store{r}$} (1);
      \path (1) edge[loop below] node[below] {$ \indata=0, \store{r}$} (1);
    \end{tikzpicture}
    \caption{A register automaton over $(\bbN,<,0)$ (state $q$ is accepting).}
    \label{fig:UpWeGo}
    \end{subfigure} \quad
    \begin{subfigure}[b]{0.45\textwidth}
    \begin{tikzpicture}[xscale=0.6,yscale=0.4]
      \node (0) at (-0.3,0) {\footnotesize 0};
      \draw[->] (0,0) -- (0,8) node[above] {\footnotesize value};
      \draw[->] (0,0) -- (8,0) node[right] {\footnotesize execution};
      \draw[dashed] (0,6) -- (7.5,6) node[right] {$m$};
      \draw (0,0) -- (1,2.5) -- (2,4) -- (2,0) -- (3,6) -- (3,0) -- (4,4) -- (5,5) -- (5,0) -- (6,1) -- (7,5.5) -- (7,0) -- (7.5,1) node[above right] {$r$};
    \end{tikzpicture}%
    \caption{An example of the sequence of values taken by registers $r$ and $m$ along the run.}
    \label{fig:AW_N}
  \end{subfigure}
  \caption{A register automaton whose action words do not form an $\omega$-regular language.}
  \label{fig:atm_AW_N}
  \end{figure}
  Its syntactical language is
  $$
    (0<\indata,\da m) \big( (r<\indata< m,\da r) \mid (0 = \indata, \da r) \big)^\omega
  $$%
  which includes not only feasible but also unfeasible action words, e.g.\ $(0<\indata,\da m) (r<\indata< m,\da r)^\omega$.
  The feasible accepted action words have the form
  $$
    (0<\indata,\da m) \prod_{i=1}^\infty\big((r<\indata< m,\da r)^{n_i} (0 = \indata, \da r)\big)
  $$%
  such that the numbers $(n_i)_i$ are uniformly bounded by some value;
  the bound corresponds to the first read data value.
  This language is not $\omega$-regular but an $\wB$-language~\cite{BC06}.
\end{example*}

\myparagraph{Register transducers}\label{sec:def:transducers}
A \emph{$k$-register transducer} is a tuple $T = (Q,q_0,R,\delta)$,
where $Q$, $q$, $R$ ($|R|=k$) are as in automata but
$\delta : Q\x\MTst\  \to  \Asgn \x R \x Q$. 
Note that $\delta$ is a total function;
moreover, since we restrict to maximal tests,
exactly one test holds per incoming data value,
so the transducers are deterministic and complete.
A configuration is a pair $(p,\v) \in Q \x \bbD^R$.
From configuration $(p,\v)$,
on reading $\d \in \bbD$,
the transducer takes the unique transition $p \trans{\tst,\asgn \mid r} q$ such that $\v,\d \models \tst$, updates its configuration to $(q,\v')$ where $\v \trans{\d,\tst,\asgn} \v'$, and outputs the value $\v'(r)$. Note that the output is produced \emph{after} assignment. We then write $(p,\v) \trans[T]{\d,\tst,\asgn \mid r, \v'(r)} (q,\v')$, or simply $(p,\v) \trans[T]{\d \mid \v'(r)} (q,\v')$.
A \emph{run} of $T$ on an input data word $\dI_0 \dI_1 \ldots$ is a sequence
$(q_0,\v_0) (q_1,\v_1)\ldots$ such that $(q_0,\v_0)$ is initial and for all $i \geq 0$, $(q_i,\v_i) \trans{\dI_i, \tst_i,\asgn_i \mid r_i, \dO_i} (q_{i+1},\v_{i+1})$ for some unique $\dO_i \in \bbD$.
The sequence $\dO_0\dO_1\ldots$ is the \emph{output word} of $T$ on reading $\dI_0\dI_1\ldots$;
since the transducers are deterministic and have a run on every input word,
the output word is uniquely defined.
The sequence $\dI_0 \dO_0 \dI_1 \dO_1 \ldots$ is called the \emph{input-output word}.
We then say that the \emph{transducer action word} $\tst_0 (\asgn_0,r_0) \tst_1 (\asgn_1,r_1) \ldots \in \big(\MTst{\cdot}(\Asgn\x R)\big)^\omega$ is \emph{feasible} by $(\v_0, \dI_0, \dO_0) (\v_1, \dI_1, \dO_1) \ldots$. It is naturally associated with the automaton action word $(\tst_0,\asgn_0) (\indata = r_0,\varnothing) (\tst_1,\asgn_1) (\indata=r_1,\varnothing)\ldots$, which is then feasible by $\v_0 \dI_0 \v_1 \dO_0 \v_1 \dI_1 \v_2 \dO_1 \dots$. The set of all transducer action words over $R$ in data domain $\D$ is denoted by $\TW_R^\D$.
The \emph{language} $L(T)$ consists of all input-output words of $T$.

A \emph{finite transducer} is a standard Mealy machine:
it is a tuple $(\Sigma,\Gamma,Q,q_0,\delta)$,
where $\Sigma$ and $\Gamma$ are finite input and output alphabets,
$\delta : Q \x \Sigma \to \Gamma \x Q$,
and the definition of language is standard.
Treating a register transducer $T$ syntactically gives a finite transducer denoted $\Ts$
of the same structure as $T$ with $\Sigma =  \MTst_R$ and $\Gamma =  \Asgn_R\x R$.

\myparagraph{Synthesis problem}
Fix a data domain $(\bbD,P,C,c_0)$.
A register transducer $T$ \emph{realises} a register automaton $S$ if $L(T) \subseteq L(S)$.
The \emph{register-bounded synthesis problem} is:
\li
\- input: $k \in \bbN$ and a URA $S$;
\- output: yes iff there exists a $k$-register transducer which realises $S$.
\il
In this paper, when the synthesis problem is decidable, we
are able to synthesise, i.e., effectively construct,  a transducer
realising the specification.
We now make two remarks.
First, notice that the number of transducer states is finite but unconstrained.
Thus, register-bounded synthesis generalises classical register-free synthesis from
(data-free) $\omega$-regular specifications.
Second, observe that
transducers are complete, and therefore produce an ouptput word on every input word.
Thus,
a specification for which some input words do not have an associated output word is unrealisable.
It is known that in the finite-alphabet case,
the refined synthesis problem of good-enough synthesis~\cite{AK20},
which requires a transducer to react only to inputs that belong to the domain of the specification,
is still decidable.
However,
the good-enough register-bounded synthesis is undecidable~\cite[Chapter~8]{ExibardThesis}.

\begin{example*}
  We illustrate the synthesis problem by describing a specification,
  its URA, and a register transducer realising it.

  Let us start with the specification of priority arbiters.
  Such an arbiter reads an ID of a process requesting the resource,
  and outputs an ID of a process to whom the resource is granted.
  The specification requires that every requesting process
  is
  either acknowledged consecutively twice on the output,
  or this is done for a process of higher ID.
  We model the specification using the URA over $(\bbN,<,0)$ with a single register
  from Figure~\ref{fig:req_grant_spec}.
  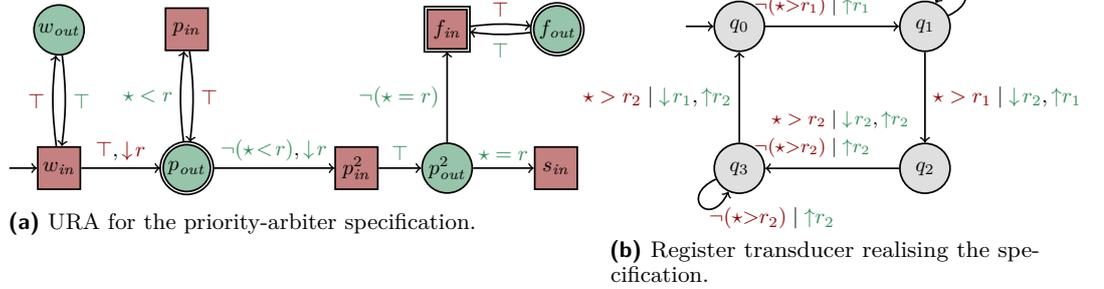
\begin{figure}[htb]
    \centering
    \begin{subfigure}[c]{0.54\textwidth}
      \centering
  \vspace{-2mm}
  \centering
  \scalebox{0.8}{
    \begin{tikzpicture}[->, thick, trim left=-7mm]  
    \node[input state, initial]                        (ii) {$w_\mathit{in}$};
    \node[output state, above=1.5cm of ii]             (io) {$w_\mathit{out}$};
    \node[output state, right= 1.3cm of ii, accepting] (po) {$p_\mathit{out}$};
    \node[input state, above=1.5cm of po]              (pi) {$p_\mathit{in}$};
    \node[input state, right= 1.99cm of po]             (si) {$p_\mathit{in}^2$};
    \node[output state, right= 0.7cm of si]             (sout) {$p_\mathit{out}^2$};

    \node[input state, right= 1cm of sout]             (success) {$s_\mathit{in}$};

    \node[input state, above= 1.5cm of sout,accepting]             (fin) {$f_\mathit{in}$};
    \node[output state, right= 1cm of fin,accepting]             (fout) {$f_\mathit{out}$};

    \path (ii) edge[bend left=7] node[left]  {$\inColored{\top}$}      (io);
    \path (io) edge[bend left=7] node[right] {$\outColored{\top}$}     (ii);
    \path (ii) edge            node[above] {$\inColored{\top}, \inColored{\store{r}}$}   (po);
    \path (po) edge[bend left=7] node[left]  {$\outColored{\indata < r}$} (pi);
    \path (pi) edge[bend left=7] node[right] {$\inColored{\top}$}      (po);
      \path (po) edge     node[above] {$\outColored{\neg (\indata\!<\!r)},\outColored{\store{r}}$}  (si);
    \path (si) edge     node[above] {$\outColored{\top}$}  (sout);
    \path (sout) edge   node[above] {$\outColored{\indata = r}$}  (success);

      \path (sout) edge   node[left] {$\outColored{\neg(\indata = r)}$}  (fin);

    \path (fout) edge[bend left=7] node[below] {$\outColored{\top}$}  (fin);
    \path (fin) edge[bend left=7]  node[above] {$\inColored{\top}$}  (fout);

    \end{tikzpicture}
    }
      \caption{\label{fig:req_grant_spec}%
        URA for the priority-arbiter specification.}
    \end{subfigure}
    ~
    \begin{subfigure}[c]{0.4\textwidth}
      \centering
  \centering
  \raisebox{-5mm}{\scalebox{0.8}{
    \begin{tikzpicture}[->, thick,trim left=-20mm]
    \node[trans state, initial]                (q0) {$q_0$};
    \node[trans state, right=2.2cm of q0]      (q1) {$q_1$};
    \node[trans state, below=1.5cm of q1]      (q2) {$q_2$};
    \node[trans state, left=2.2cm of q2]       (q3) {$q_3$};

    \path
      (q0) edge
        node[above]
        {
          \specialcellC
          {
            $\ \ \,\inColored{\indata > r_1} \mid
             \outColored{\store{r_1}}, \outColored{{\uparrow}r_1}$\\
            $\hspace{-6.5mm}\inColored{\neg (\indata{>}r_1)} \mid \outColored{{\uparrow}r_1}$
          }
        }
      (q1);

    \path[out=25,in=65,looseness=6]
      (q1) edge
        node[right]
        {
          $\inColored{\neg(\indata{>}r_1)} \mid \outColored{{\uparrow}r_1}$
        }
      (q1);

    \path
      (q1) edge
        node[right]
          {
            $\inColored{\indata > r_1} \mid
             \outColored{\store{r_2}}, \outColored{{\uparrow}r_1}$
          }
      (q2);

    \path
      (q2) edge
        node[above]
        {
          \specialcellC
          {
            $\ \ \,\inColored{\indata > r_2} \mid
              \outColored{\store{r_2}}, \outColored{{\uparrow}r_2}$\\
            $\hspace{-6.5mm}\inColored{\neg(\indata{>}r_2)}\mid \outColored{{\uparrow}r_2}$
          }
        }
      (q3);

    \path[out=205,in=245,looseness=6]
      (q3) edge
        node[right,yshift=-2mm]
        {
          $\inColored{\neg(\indata{>}r_2)} \mid \outColored{{\uparrow}r_2}$
        }
      (q3);

    \path
      (q3) edge
      node[left]
        {
          $\inColored{\indata > r_2} \mid
           \outColored{\store{r_1}}, \outColored{{\uparrow}r_2}$
        }
      (q0);
    \end{tikzpicture}
  }
  }
  \vspace{-2mm}
      \caption{\label{fig:transd_req_grant}%
        Register transducer realising the specification.}
    \end{subfigure}
    \caption{A URA specification and a transducer implementing it.}
  \end{figure}
\end{example*}
  The automaton reads words interleaving between arbiter data input and output,
  so its states are partitioned into box states (for reading input) and circle states (for reading output).
  The double-circle states are rejecting and can be visited only finitely often.
  Thus, a run looping in wait states $w_\mathit{in}$ and $w_\mathit{out}$ is accepting.
  Branching is universal,
  hence some run always loops around $w_\mathit{in}$ and $w_\mathit{out}$.
  On reading an ID of a requesting process,
  a copy of the automaton moves from $w_\mathit{in}$ to a pending state $p_\mathit{out}$
  while storing the ID into register $r$.
  It stays in states $p_\mathit{out}$ and $p_\mathit{in}$ as long as the request is not acknowledged,
  and such an infinite run is rejecting.
  If the request is eventually acknowledged (transitions from $p_\mathit{out}$ to a sink state $s_\mathit{in}$),
  the run dies, so it is accepting.
  If a run reaches the failure state $f_\mathit{in}$, it is rejecting.

  Figure~\ref{fig:transd_req_grant} depicts a transducer with two registers $r_1$ and $r_2$
  realising the above specification. On the left of the vertical bar are the tests over the inputs received by the transducer (in red), and on the right is the output action performed by the transducer (in green).
  For example, from state $q_0$ to $q_1$,
  if the input data $\d$ is larger than the data stored in register $r_1$,
  the transducer stores it into $r_1$, and outputs the content of $r_1$.
  The transducer uses one register to store the maximal value seen so far,
  while outputting the content of the other register,
  and the roles of these registers interchange as the transducer transits along the states.
  Thus, the instance of register-bounded synthesis with the described URA and $k=2$ has a positive answer.
  However, when $k=1$ the answer is negative.

\section{Sufficient Condition for Decidable Synthesis for URA}
\label{sec:dec-recipe}

In this section, we first show a reduction from register-bounded
synthesis to (register-free) finite-alphabet synthesis. In the following, we fix a data domain $\D$. Given a specification $S$ (as a URA over $\D$) and a bound $k$, we show how to
construct a finite-alphabet specification $\Wf$ on action words over $k$ registers,
which is realisable by a finite-alphabet
transducer iff $S$ is realisable by a $k$-register transducer
(Lemma~\ref{lem:w-wf}). The main idea is to see the actions of the URA
and of the sought $k$-register transducer as finite-alphabet letters.
In particular,
the specification $\Wf$
accepts a transducer action word $\ak$
iff
every action word $\aS$ of the specification $S$,
such that both $\ak$ and $\aS$ are feasible by the same data word,
is accepted by $\Ss$.
One can compose automata and transducer action words through a form of parallel product, which allows to talk about their joint feasibility.
Then, in general,
$\Wf$ is not necessarily $\omega$-regular, and in a second step,
we provide sufficient conditions on the data domain making
synthesis wrt.\ $\Wf$ decidable, namely, that it can be
under-approximated by an $\omega$-regular language which coincides with $W_{S,k}^f$
over lasso words (Section~\ref{subsec:gen}). We obtain a general
decidability result for data domains having this property
(Theorem~\ref{thm:gendec}).
We then instantiate this result for data domain
$(\mathbb{N},<,0)$ (Section~\ref{subsec:overn}).

In the following, we fix a URA $S$ with registers $R_S$ and a disjoint set $R_k$ consisting of $k$ registers, and let $R = R_S \uplus R_k$.
Given a transducer action word
$\ak = \tst^k_0\,(\asgn^k_0,r^k_0) \ldots \in \TW_{R_k}^\D$
and
an automaton action word
$\aS = (\tstSI_0,\asgnSI_0) (\tstSO_0,\asgnSO_0) \ldots \in \AW_{R_S}^\D$,
the product $\ak \otimes \aS$ of $\ak$ and $\aS$ is the automaton action word over registers $R$ defined as $(\tst^k_0 \wedge \tstSI_0, \asgn^k_0 \cup \asgnSI_0) ((\indata = r^k_0) \wedge \tstSO_0,\asgnSO_0) \dots$,
which is essentially the parallel product of $\aS$ and of the automaton word associated with $\ak$.
We now show how to abstract a data specification given as URA $S$ with registers $R_S$
by a \emph{finite-alphabet} specification over $k$-register transducer action words.
Let $\F_R^\D$ be the set of automata action words over $R$ feasible in $\D$,
then we define
$$
  \Wf =
    \big\{
      \a_k \in \TW_{R_k} \mid
      \forall \a_S \in \AW_{R_S}\:
      \ak \otimes \aS \in \F_R^\D ~\Impl~ \aS \in L(\Ss)
    \big\}.
$$%
Thus,
$\Wf$ \emph{rejects} a feasible transducer action word $\ak$ iff there is an automaton action word $\aS$ feasible by the same data word as $\ak$ and rejected by $S$.

\begin{lemma}\label{lem:w-wf}
  These two are equivalent:
  \li
  \- a URA $S$ is realisable by a $k$-register transducer,
  \- $\Wf$ is realisable (by a finite-alphabet transducer).
  \il
\end{lemma}
\begin{proof}
  $\Impl$: Assume that $S$ is realisable by a register transducer $T$, i.e. $L_\D(T) \subseteq L_\D(S)$. Let $\ak \in L(\Ts)$, and let $\aS \in \AW_{R_S}$ such that $\ak \otimes \aS \in \F_R^\D$. Then, $\ak \otimes \aS$ is feasible by some input-output data word $w = \dI_0 \dO_0 \dI_1 \dO_1 \dots$. By definition of the product, both $\ak$ and $\aS$ are feasible by $w$. Since $L_\D(T) \subseteq L_\D(S)$, if $\aS$ labels a run of $S$ on $w$, it means that it is accepting otherwise $w \notin L_\D(S)$ since $S$ is a universal automaton. Thus, $\aS \in L(\Ss)$.

  $\Implied$: Conversely, assume that $\Wf$ is realisable by some finite transducer $M$, and let $T$ be the associated register transducer, i.e. such that $\Ts = M$. Let $w \in L_\D(T)$ and let $\ak$ be the action word labelling the run of $T$ on $w$. Let $\aS$ be an action word labelling a run of $S$ on $w$ if it exists (it might be that $w$ is accepted by $S$ by having no run on it). Then, $\ak \otimes \aS$ is feasible by $w$. By definition of $\Wf$, it means that $\aS \in L(\Ss)$, so $\aS$ labels an accepting run of $S$ on $w$. Overall, all runs of $S$ on $w$ are accepting, so $w \in L_\D(S)$. Thus, $L_\D(T) \subseteq L_\D(S)$, i.e. $T$ realises $S$.
\end{proof}

\subsection{General Decidability Result}\label{subsec:gen}

In $(\bbN,<,0)$, $\Wf$ is not $\omega$-regular in general. 
To overcome this obstacle,
we define the notion of $\omega$-regularly approximable data domains.
Such domains have an $\omega$-regular equi-realisable subset of $\Wf$.

Let $\lasso_R$ be the set of lasso-shaped\footnote{A word $w$ is \emph{lasso-shaped} (or regular, or ultimately periodic) if it is of the form $w = u v^\omega$ for some finite words $u$ and $v$.} action words over a given set of registers $R$;
we write $\lasso$ when $R$ is clear.
A data domain $\bbD$ is \emph{$\omega$-regularly approximable} (\emph{regapprox})
if for every $R$
there exists an $\omega$-regular language $\QF_R \subseteq (\Tst_R\x\Asgn_R)^\omega$
satisfying
$$
  \QF_R \cap\lasso_R ~\subseteq~ \F_R ~\subseteq~ \QF_R
$$%
and recognisable by a nondeterministic B\"uchi automaton
that can be effectively constructed given $R$.
The definition implies that $\F_R$ and $\QF_R$ coincide on lasso words.
Such a set $\QF_R$ is called \emph{regular approximation} and written as $\QF$ when $R$ is clear.

\begin{example*}
  The data domains $(\bbD,=)$ and $(\bbQ,<)$ are regapprox
  because their sets $\F_R$ for every $R$ are $\omega$-regular,
  so there is no need to approximate them.
  On these domains, to check whether a given action word is feasible,
  one can track the relations between the registers and check if the read tests are consistent with these relations.
  For instance,
  if $r_1 < r_2$ but we read the test $* = r_1 = r_2$, then the action word is unfeasible.

  The domain $(\bbN,<,0)$ is also regapprox.
  Here, it is not sufficient to track the relations between the registers.
  We also need to ensure that between any two stored data values only a bounded number of different values is inserted along the action word.
  (Recall the example on page~\pageref{page:example:UpWeGo} with Figure~\ref{fig:UpWeGo}.)
  However, when an action word is lasso-shaped,
  it suffices to check the absence of an \emph{infinite} number of such insertions.
  The latter can be checked by an $\omega$-regular automaton,
  which allows for proving the regapproximability of $(\bbN,<,0)$.

  Finally, consider the data domain $(\bbN,\{S,=\},\{0\},0)$,
  where $S$ is the successor relation, i.e.\ $S(a,b)$ holds iff $a = b+1$.
  This domain is not regapprox.
  Intuitively, this is because the domain allows for counting,
  which enables non $\omega$-regular phenomena even in lasso words.
  We prove this by contradiction.
  Consider the following $\omega$-regular language of action words over a single register $r$:
  $$
    L ~=~ \big\{\big(S(*,r), \da r\big)^n \big(S(r,*), \da r\big)^m \big(*=0=r, \emptyset\big)^\omega \mid n,m\in\bbN \big\},
  $$
  i.e.\ the value in $r$ is incremented $n$ times, then decremented $m$ times,
  then compared to zero and not updated.
  $L$ contains feasible as well as unfeasible action words.
  Every feasible word of $L$ has $n=m$, hence $\F\cap L$ is not $\omega$-regular.
  Moreover, every word of $L$ is a lasso, thus $L \cap \lasso = L$.
  Let us assume that the data domain is regapprox, witnessed by $\QF$ for $R = \{r\}$.
  Since $\QF \cap \lasso = \F \cap \lasso$ by definition,
  we get
  $$
    \QF \cap L = \QF \cap \lasso \cap L = \F \cap \lasso \cap L = \F \cap L.
  $$
  The language $\QF\cap L$ is $\omega$-regular, but $\F \cap L$ is not. Contradiction.
  Therefore $(\bbN,\{S,=\},\{0\},0)$ is not regapprox.
  \qed
\end{example*}

Given a URA $S$ with registers $R_S$ and $k$, we define
$$
  \Wqf =
    \big\{
      \a_k \mid
      \forall \a_S\:
      \ak \otimes \aS \in \QF_R ~\Impl~ \aS \in L(\Ss)
    \big\},
$$%
where $R = R_S \uplus R_k$.
The definition of $\Wqf$ differs from $\Wf$ only in using $\QF_R$ instead of $\F_R$.
Since $\F_R\subseteq\QF_R$, we have $\Wqf\subseteq \Wf$.

We now show that $\Wqf$ is $\omega$-regular
(which essentially follows from $\omega$-regularity of $\QF$ and $\Ss$),
and estimate the size of an automaton recognising $\Wqf$ and the time needed to construct it.
For that we use the following terminology for functions of asymptotic growth:
a function is $poly(t)$ if it is $O(t^\kappa)$,
$exp(t)$ if it is $O(2^{t^\kappa})$, and
$2exp(t)$ if it is $O(2^{2^{t^\kappa}})$,
for a constant $\kappa\in\bbN$.
When $poly$, $exp$, and $2exp$ are used with several arguments, the maximal among them shall be taken for $t$. The construction and complexity analysis rely on standard automata techniques; we refer to Appendix~\ref{app:lem:Wqf-is-regular} for details.
\begin{lemma}\label{lem:Wqf-is-regular}
  Let $S$ be a URA and let $k \geq 1$. Then, $\Wqf$ is $\omega$-regular.
  Moreover,
  $\Wqf$ is recognisable by a universal co-B\"uchi automaton with $O(2^kNnc)$ many states
  that can be constructed in time $poly(N,n,exp(r,k))$,
  where $n$, $r$, and $c$ are the number of states, registers, and priorities in $S$,
  and $N$ is the number of states in a nondeterministic B\"uchi automaton
  recognising $\QF_{R_S\uplus R_k}$.
\end{lemma}

We now prove that $\Wf$ and $\Wqf$ are equi-realisable.
For $\omega$-regular specifications (like $\Wqf$)
there is no distinction between realisability by finite- and infinite-state transducers~\cite{BL69}.
This is not known for $\Wf$ specifications over domains such as $(\bbN,<,0)$; 
we leave this question for future work,
and in this paper focus on realisability by \emph{finite-state} transducers.

\begin{lemma}\label{lem:Wf-Wqf}
  $\Wf$ is realisable by a finite-state transducer iff $\Wqf$ is realisable by a finite-state transducer.
\end{lemma}
\begin{proof}
  Direction $\Implied$ follows from the inclusion $\F \subseteq \QF$, which implies $\Wqf\subseteq\Wf$.
  Consider direction $\Impl$.
  Let $T$ be a finite-state transducer that $T$ does not realise $\Wqf$.
  We show that $T$ does not realise $\Wf$ either.
  First, we have that
  $L(T) \not \subseteq \Wqf$,
  so the language $\{\ak \otimes \aS \in \AW_R^\D \mid \ak \in L(T) \land \ak \otimes \aS \in \QF \land \aS \notin L(\Ss)\}$ is nonempty.
  Since $\QF$ and $L(\Ss)$ are $\omega$-regular,
  and since $T$ is a finite-state transducer,
  this language is $\omega$-regular.
  Thus, it contains a
lasso-shaped word $\ak \otimes \aS$; by definition of the product,
both $\ak$ and $\aS$ are then lasso-shaped. Since
$\QF\cap\lasso\subseteq \F$, we get that $\aS$ is feasible, i.e. $\ak \otimes \aS \in \{\ak \otimes \aS \mid \ak \in L(T) \wedge \ak \otimes \aS \in \F \wedge \aS \notin L(\Ss)\}$, which implies that $L(T) \not \subseteq \Wf$: $T$ does not realise $\Wf$.
\end{proof}

We are now able to prove the main result of this paper.

\begin{theorem}\label{thm:gendec}
    Let $\D$ be a regapprox data domain such that for every
    set of registers $R$, one can construct a nondeterministic
    B\"uchi automaton with $n_\textsf{QF}$ states recognising $\QF_R$ in time
    $f(|R|)$ for some function $f$. Then:
    \li
    \- register-bounded synthesis for URAs over $\D$
    is decidable in time $exp(exp(k,r),n_\textsf{QF},n,c)+f(k+r)$,
       where
       $n$ is the
       number of states of the URA,
       $c$ its number of priorities,
       $r$ its number of registers,
       $k$ is the number of transducer registers.
       It is \textsc{ExpTime}-c for fixed $r$ and $k$.

    \- For every positive instance of the
       register-bounded synthesis problem, one can
       construct, within the same time complexities, a register transducer realising the specification.
    \il
\end{theorem}

\begin{proof}
    Lemmas~\ref{lem:w-wf},\ref{lem:Wqf-is-regular},\ref{lem:Wf-Wqf}
    reduce register-bounded synthesis to (finite-alphabet)
    synthesis for the $\omega$-regular specification $\Wqf$.
    Since synthesis wrt.\ to $\omega$-regular specifications
    is decidable, we get the decidability part of the theorem. Let us
    now study the complexity. Let $R_S$ be the set of $r$ registers of
    the URA and $R_k$ be a disjoint set of $k$ registers.
    First, one needs to construct an
    automaton recognising $\QF_{R_S\cup R_k}$. This is done by
    assumption in time $f(k+r)$. Then, one can apply
    Lemma~\ref{lem:Wqf-is-regular} and get that $\Wqf$
    can be recognised by universal co-B\"uchi automaton $A$ with
    $O(2^kn_\mathit{qf}nc)$ states, which can be constructed in time
    $poly(n_\mathit{qf},n,exp(r,k))$.
    A universal co-B\"uchi automaton with $m$ states can be determinised into
    a parity automaton with $exp(m)$ states and $poly(m)$ priorities (see e.g.\ \cite{Pit06}).
    Recall that the alphabet of $A$ is $\Tst_k\cup (\Asgn_k\times
    R_k)$. Hence by determinising $A$,
    and seeing it as a two-player game arena,
    we get a parity game with $exp(k)$ edges (corresponding to the
    actions of Adam and Eve), $exp(exp(k),n_\textsf{QF},n,c))$ states, and $poly(exp(k),n_\textsf{QF},n,c))$ priorities.
    The latter can be solved in polynomial time in the number of its states,
    as the number of priorities is logarithmic in the number of states (see e.g.\ \cite{CJKLS17}),
    giving the overall time complexity $exp(exp(k),n_\textsf{QF},n,c))$ for
    solving the game. If we sum this to the complexity of constructing
    an automaton for $\Wqf$ plus the complexity for construction an
    automaton for $\QF$, we get $exp(exp(k),n_\textsf{QF},n,c)) +
    poly(n_\textsf{QF},n,exp(r,k)) + f(r+k)$, which is
    $exp(exp(k,r),n_\textsf{QF},n,c)) + f(r+k)$. If both $r$ and $k$ are
    fixed, then $exp(k,r)$ and $f(r+k)$ are constants, so the
    complexity is exponential only. It is folklore that the hardness
    holds in the register-free setting (for $r = k = 0$). See for
    example~\cite[Proposition~6]{DBLP:conf/icalp/FiliotJLW16}  for a
    proof in the finite word setting over a
    finite alphabet (which straightforwardly generalises to infinite
    words). There, the proof is done for nondeterministic finite
    automata, but by determinacy, hardness also holds for universal automata, as they are dual.

    Now, if a URA specification is realisable for some given $k$, then
    by Lemmas~\ref{lem:w-wf} and \ref{lem:Wf-Wqf}, $\Wqf$ is realisable
    by a finite-alphabet transducer $M$. Since $\Wqf\subseteq \Wf$,
    $M$ also realises the specification $\Wf$. The mapping
    $\cdot_{synt}$ which turns a register transducer into a
    finite-alphabet transducer is bijective, and hence there exists a
    register transducer $T$ such that $T_{synt} = M$.
    The proof of Lemma~\ref{lem:w-wf} exactly shows that $T$ realises
    $S$, hence we are done.
\end{proof}

\subsection{Register-bounded Synthesis over Data Domain \texorpdfstring{$(\bbN,<,0)$}{(N,<,0)}}\label{subsec:overn}

We instantiate Theorem~\ref{thm:gendec} for the data domain $(\bbN,<,0)$.
In~\cite{EFK21},
though there was no general notion of $\omega$-regular approximability for data domains,
it was implicitly used for $(\bbN,<,0)$.
The following fact follows from~\cite[Thm.8]{EFK21} after adapting to our notions.\footnote{%
  Strictly speaking, their paper considers maximal tests only.
  However, using their deterministic automaton for $\QF_R$ over action words with maximal tests,
  we can construct a \emph{non}det.\ automaton recognising quasi-feasible action words with all tests, incl.\ partial ones.
  Our nondet.\ automaton, on reading a partial test,
  \emph{guesses} its completion into a maximal test and simulates the original automaton on it.}
\begin{fact}\label{lem:N-goodness-atm}
  For all $R$, $\DN$ has a witness $\QF_R$ of $\omega$-regular approximability expressible by a nondeterministic parity automaton
  with $exp(|R|)$ states and $poly(|R|)$ priorities,
  which can be constructed in time $exp(|R|)$.
\end{fact}
A parity automaton can be translated to a nondeterministic B\"uchi automaton with a quadratic number of states, so we can instantiate Theorem~\ref{thm:gendec} on domain $(\bbN,<,0)$ and get:
\begin{theorem}
  \label{thm:dec_synt_N}
  For a URA in $(\bbN,<,0)$ with $r$ registers, $n$ states, and $c$ priorities,
  $k$-register-bounded synthesis is solvable in time $exp(exp(r,k),n,c)$:
  it is singly exponential in $n$ and $c$, and doubly exponential in
  $r$ and $k$. It is \textsc{ExpTime-c} for fixed $k$ and $r$.
\end{theorem}

\section{Reducibility Between Data Domains}
\label{sec:applications}

Theorem~\ref{thm:dec_synt_N} relies on the study of feasibility of action words in
$(\bbN,<,0)$ of~\cite{EFK21}, which requires some effort. Such a study could in
principle be generalised to domains such as $\bbZ$-tuples, as well as
finite strings with the prefix relation, by
leveraging the results of~\cite{DD16}. However, this would come at the price of
a high level of technicality. We choose a different path, and introduce a notion
of reducibility between domains, which allows us to reuse the study of $(\bbN,<,0)$
and yields a compositional proof of the decidability of register-bounded
synthesis for the quoted domains.

\begin{definition*} 
  A data domain $\D$ \emph{reduces} to a data domain $\D'$ if
  for every finite set of registers $R$,
  there exists a finite set of registers $R'$ and
  a rational relation\footnote{Given two finite alphabets $\Sigma$ and $\Gamma$, a
  relation $K \subseteq \Sigma^\omega \times \Gamma^\omega$ is rational if there
  exists an $\omega$-regular language $L \subseteq (\Sigma \cup \Gamma)^\omega$
  such that $K = \{(\mathrm{proj}_{\Sigma}(u),\mathrm{proj}_{\Gamma}(u)) \mid u \in L\}$. This is
  equivalent to saying that it
  can be computed by a nondeterministic asynchronous finite-state transducer
  over input $\Sigma$ with output in $\Gamma^*$. See, e.g., \cite[Section
  3]{DBLP:books/lib/Berstel79}.} $K$ between
  $R$-automata action words in $\D$ and
  $R'$-automata action words in $\D'$ that \emph{preserves feasibility},
  in the sense that
  for every $R$-action word $\a \in (\Tst^{\D}_R \Asgn_R)^\omega$:
  $\a$ is feasible in $\D$ iff
  there exists an $R'$-action word in $K(\a)\in (\Tst^{\D'}_{R'}
  \Asgn_{R'})^\omega$ feasible in $\D'$.\footnote{Note that we do
    not forbid the existence of unfeasible action words in the image.}
\end{definition*}
\begin{remark*}
  Reducibility is a transitive relation, since rational relations
  are closed under composition~\cite[Theorem 4.4]{DBLP:books/lib/Berstel79}, and
  feasibility preservation is transitive.
\end{remark*}

Since $K$ is rational and preserves feasibility,
for all $R$, $K^{-1}(\QF_{R'})$ is a witness of regapproximability, where $R'$ is as in
the above definition
(see the proof below for details),
thus we get:
\begin{lemma}\label{lem:dectransfer}
    If $\D$ reduces to $\D'$ and $\D'$ is regapprox,
    then $\D$ is regapprox.
\end{lemma}
\begin{proof}
    Let $R$ be a fixed set of registers, and let $R'$ be a set of registers satisfying the definition of reducibility.
    Let $\F$ (respectively, $\F'$) be the set of $R$-action words feasible in $\D$
    (resp., feasible $R'$-action words in $\D'$).

    Our goal is to define an $\omega$-regular set $\QF$ (for $R$)
    s.t.\ $\QF\cap lasso \subseteq \F \subseteq \QF$. 
    Since $\D'$ is regapprox,
    there is an $\omega$-regular set $\QF'$ (for $R'$) s.t.\ $\QF' \cap\ lasso \subseteq \F' \subseteq \QF'$.
    Define $\QF = K^{-1}(\QF')$; as the preimage of an $\omega$-regular set
    by a rational relation, it is (effectively) $\omega$-regular,
    thus satisfying one of the condition for $\D$ to be regapprox.

    We now show that $\F \subseteq \QF$.
    Before proceeding, notice that $\F = K^{-1}(\F')$, since $K$
    preserves feasibility.
    Since $\F'\subseteq \QF'$, we have $K^{-1}(\F')\subseteq K^{-1}(\QF')$,
    hence $\F \subseteq \QF$.

    It remains to show that $\QF \cap lasso\subseteq \F$.
    The inclusion $\QF'\cap lasso\subseteq \F'$ implies
    $K^{-1}(\QF'\cap lasso)\subseteq K^{-1}(\F') = \F$ (the latter equality is because $\F = K^{-1}(\F')$).
    We prove that $\QF \cap lasso\subseteq K^{-1}(\QF'\cap lasso)$,
    which entails the desired result.
    Pick an arbitrary $\a \in \QF\cap lasso$.
    Since $K$ is rational, $K(\a)$ is $\omega$-regular.
    Moreover, $\QF'$ is $\omega$-regular,
    which entails that $K(\a)\cap \QF'$ is $\omega$-regular as well.
    Since $\a \in K^{-1}(\QF')$, the intersection $K(\a) \cap \QF'$ is nonempty.
    Since $K(\a) \cap \QF'$ is $\omega$-regular and nonempty,
    it contains a lasso word $\a'$.
    Thus, $\a'\in K(\a) \cap \QF' \cap lasso$,
    hence $\a \in K^{-1}(\QF'\cap lasso)$.
\end{proof}

As a direct consequence of Lemma~\ref{lem:dectransfer} and
Theorem~\ref{thm:gendec}, we get the following result:
\begin{theorem} \label{thm:reducibility_synthesis}
  If a data domain $\D$ reduces to a regapprox data domain,
  then register-bounded synthesis is decidable for $\D$. Moreover, for any
  positive instance of the register-bounded synthesis problem over
  $\D$, one can effectively construct a register transducer
  realising the specification of that instance.
\end{theorem}

\subsection{Adding Labels to Data Values}
As a first application, we show that one can equip data values with labels
from a finite alphabet while preserving regapproximability. By
Theorem~\ref{thm:reducibility_synthesis}, this yields decidability of
register-bounded synthesis for such domains.

Formally, given a data domain $\D = (\bbD, P, C, c_0)$ and a finite alphabet
$\Sigma$, we define the domain of $\Sigma$-labeled data values over
$\D$ as
$\Sigma \times \D = (\Sigma \times \bbD, P \cup \{\lab_{\sigma} \mid \sigma \in \Sigma\}, \Sigma \times C, (\sigma_0,c_0))$,
where $\sigma_0 \in \Sigma$ is a fixed but arbitrary element of $\Sigma$ and,
for each $\sigma \in \Sigma$, $\lab_{\sigma}(\gamma,d)$ holds if and only if
$\gamma = \sigma$.

\begin{lemma} \label{lem:lab_red_nolab}
  For all finite alphabet $\Sigma$ and data domain $\D$,
  $\Sigma \x \D$ reduces to $\D$.
\end{lemma}
\begin{proof}
    Wlog we assume that the set of constants $C$ is the singleton $\{c_0\}$ (modulo
    adding new predicates to $P$). Let $\Sigma =
    \{\sigma_0,\sigma_1,\dots,\sigma_n\}$, where $\sigma_0$ is such that
    $(\sigma_0,c_0)$ is the initialiser of $\Sigma \times \D$. We first define an
    encoding at the level of data words. Let $\mu : \Sigma\rightarrow
    \bbD$ be an injective mapping such that $\mu(\sigma_0) = c_0$.  A data word
    $u$ over $\D$ is a
    $\mu$-encoding of $v = (\sigma_{i_1},\d_1)(\sigma_{i_2},\d_2)\dots$ if it is
    equal to $\mu(\sigma_1)\dots \mu(\sigma_n)
    \mu(\sigma_{i_1})\d_1\mu(\sigma_{i_2}) \d_2\dots.$ The data word $u$ is a
    \emph{valid encoding} of $v$ if it is a $\mu$-encoding of $v$ for some $\mu$.

    Now, the idea is to define a rational relation $K$ from action
    words $\a$ over $\Sigma \x \D$ to actions words $\overline{b}$ over $\D$ such
    that $\a$ is feasible by some $u$ iff there exists
    $\overline{b}$ such that $(\a,\overline{b})\in K$
    and $\overline{b}$ is feasible by a valid encoding of $u$.
    Let $R$ be a set of registers and assume $\a$ is built
    over $R$. Let $R' = \{r_\sigma\mid \sigma\in \Sigma\}\uplus
    R$. Then, any $\overline{b}$ such that
    $(\a,\overline{b})\in K$ should ensure that
    the $n$ first data values are distinct and store them in
    $r_{\sigma_1},\dots,r_{\sigma_n}$ respectively.
    So, we require
    that $\overline{b}$ is of the form
    $\overline{b} = b_\Sigma \cdot b_{\a}$
    where
    $b_\Sigma = (\tst_{\sigma_1}, \da r_{\sigma_1}) \dots  (\tst_{\sigma_n}, \da r_{\sigma_n})$
    such that for all $1\leq i \leq n$,
    $\tst_i = \bigwedge_{1 \leq j \leq i} \indata \neq r_{\sigma_j}$.
    The second part $b_{\a}$ is an encoding of the
    tests and assignments of $\a = (\tst_0,\asgn_0)(\tst_1,\asgn_1)\dots$.
    It is of the form
    $b_{\a} = (\tst^{lab}_0,\emptyset)(\tst^{data}_0,\asgn_0)(\tst_1^{lab},\emptyset)(\tst^{data}_1,\asgn_1)\dots$,
    where for all $i\geq 0$:
    \li
    \- for every predicate $p\in P$ of arity $n$,
       for every $x_1,\dots,x_n\in R\cup \{\star\}$:
       if $(\neg)p(x_1,\dots,x_n)\in \tst_i$,
       then $(\neg) p(x_1,\dots,x_n)\in \tst_i^{data}$, and
    \- for all $\sigma\in\Sigma$ and $x\in R\cup\{\star\}$:
       $\lab_\sigma(x) \in \tst_i^{lab}$ \,iff\, $(r_\sigma = x) \in \tst_i$.
    \il
    Correctness follows from the construction; see
    Appendix~\ref{app:applications} for details.
\end{proof}

The latter result combined with
Theorem~\ref{thm:reducibility_synthesis} yields:

\begin{corollary}
    Let $\D$ be an regapprox data domain and $\Sigma$ be a
    finite alphabet, then register-bounded synthesis
    is decidable for $\Sigma\times \D$. 
\end{corollary}


\subsection{Quantifier-Free Interpretations}
\label{sec:QF-interpretation}
When the relation between valuations over $\D$ and over $\D'$ is local, it is more convenient to operate directly at the level of tests. To that end, we define a notion of quantifier-free interpretation (see~\cite[Section~12.3.6]{ExibardThesis} for a presentation of the notion in the context of data words), that allows us to encode elements of $\D$ as tuples of elements of $\D'$.

A \emph{quantifier-free interpretation} (or \emph{interpretation} for short) of dimension $l \geq 1$ with signature $(P,C)$ over a data domain $\D' = (\bbD', P', C')$ is given by quantifier-free formulas over signature $(P',C')$. The formula $\phi_{\text{domain}}(x_1,\ldots,x_l)$ defines the domain $\bbD=\{(\d_1,\ldots,\d_l) \mid \D' \models \phi_{\text{domain}}(\d_1,\ldots,\d_l)\}$. Then, for each constant symbol $c \in C$, the formula $\phi_{c}(x_1,\ldots,x_l)$ defines the encodings\footnote{Note that we do not assume the encoding to be unique.} of $c$ as the tuples $(\d^c_1,\ldots,\d^c_l)\in \bbD$ that satisfy $\phi_c$, i.e.\ such that $\D' \models \phi_{c}(\d^c_1,\ldots,\d^c_l)$. Finally, for each predicate $p \in P$ of arity $a$ (including $=$), the formula $\phi_{p}(x_1^1,\ldots,x_l^1,\ldots,x_1^a,\ldots,x_l^a)$ defines the predicate $p^{\D}=\left\{(\d_1^1,\ldots,\d_l^1,\ldots,\d_1^a,\ldots,\d_l^a)\,\middle|\, \D' \models \phi_{R}(\d_1^1,\ldots,\d_l^1,\ldots,\d_1^a,\ldots,\d_l^a)\right\}$.
\begin{lemma}
  \label{lem:Z_QFI_N}
  $(\mathbb{Z},<,0)$ can be defined as a 2-dimensional interpretation of $(\mathbb{N},<,0)$.
\end{lemma}
\begin{proof}
  The encoding consists of two copies of $\bbN$,
  one for positive and one for negative integers,
  whose order is reversed.
  Formally,
  $\phi_{domain}(x_1,x_2) \coloneqq x_1 = 0 \lor x_2 = 0$.
  Then, $\phi_0(x_1,x_2) \coloneqq x_1 = 0 \land x_2 = 0$;
  $\phi_=((x_1,x_2),(y_1,y_2)) \coloneqq x_1 = y_1 \land x_2 = y_2$ and
  $\phi_<((x_1,x_2),(y_1,y_2)) \coloneqq
   (x_2 = y_2 = 0 \land x_1 < y_1) \lor
   (x_1 = y_1 = 0 \land x_2 > y_2) \lor
   (x_1 = 0 \land y_1 > 0)$.
   Then, $(\bbZ,<,0)$ is isomorphic to this structure,
   through the bijection $n \geq 0 \mapsto (n,0)$ and $n < 0 \mapsto (0,-n)$.
\end{proof}

More generally, $d$-uples of integers can be easily encoded. In the following, we fix $d \geq 1$.
For $(n_1,...,n_d),(m_1,...,m_d) \in \bbZ^d$,
define $(n_1,...,n_d) <^d (m_1,...,m_d)$ iff for all $i \in \{1,...,d\}$,
$n_i \leq m_i$ and $n_j < m_j$ for some $j \in \{1,\dots,d\}$; it is a partial
order on $\bbZ^d$. The predicate $=^d$ is defined as expected.
\begin{lemma}
  \label{lem:Zk_QFI_Z}
  $(\bbZ^d,=^d,<^d,0^d)$ can be defined as a $d$-dimensional interpretation of $(\bbZ,<,0)$.
\end{lemma}
\begin{proof}
  Any tuple belongs to the domain, so we let $\phi_{\text{domain}} \coloneqq \top$. Then, $\phi_0(x_1,\dots,x_d) \coloneqq \bigwedge_{1 \leq i \leq d} x_i = 0$, $\phi_=((x_1,\dots,x_d),(y_1,\dots,y_d)) \coloneqq \bigwedge_{1 \leq i \leq d} x_i = y_i$, and similarly for $\phi_<$.
\end{proof}

The following theorem allows us to lift our results to the two domains above:
\begin{theorem}
  \label{thm:QFI_red}
  If $\D$ is a quantifier-free interpretation over $\D'$, then $\D$ reduces to $\D'$.
\end{theorem}
\begin{proof}[Proof (Sketch)]
  We outline the proof, and refer to Appendix~\ref{app:thm:QFI_red} for details. Let $\D' = (\bbD',P',C')$ be a data domain, and $\D$ be an interpretation over $\D'$ of dimension $l \geq 1$ with signature $(P,C)$. The main idea is, given a set of registers $R$, to consider $l$ copies of this set, meant to store each dimension of the interpretation. We also add $l$ copies of $C$ to store the encoding of constants, and, since tests are conducted before assignment, $l$ registers to store each component of the input tuple. Overall, an action word $\a$ over $R$ is sent to one over $(R \cup C \cup \{d\}) \times \{1,\dots,l\}$, where $d$ is a fresh register variable. Then we construct the sought relation $K$ as follows: first, it prefixes its image with a sequence of actions that store the encoding of constants in the corresponding registers, check that they indeed satisfy their respective $\phi_{c}$, and ensure that all registers in $R \times \{1,\dots,l\}$ are initialised with the encoding of $c_0$. Note that the formulas are not necessarily conjuncts, so we put them in disjunctive normal form and consider all tests that are conjuncts of the DNF. Then, each action is processed separately: an action $(\tst,\asgn)$ of $\a$ is associated with a sequence of $2l+1$ actions that consist in reading each component of the input data value $\indata$, store it in the corresponding copy of $d$, check that $\indata$ indeed belongs to the domain ($\phi_{\text{domain}}$), and that it satisfies $\tst$ (using the $(\phi_p)_{p \in P}$ to encode the predicates). Again, this implies converting the formulas in DNF, so a given action is in general associated with multiple ones. Since $K$ consists in adding a prefix and then processing each action separately, it is rational. Moreover, it preserves feasibility; more precisely for any action word $\a$, each of its corresponding data word can be associated with its encoding in $K(\a)$.
\end{proof}

By Theorems~\ref{thm:dec_synt_N},~\ref{thm:QFI_red} and~\ref{thm:reducibility_synthesis}, as well as Lemma~\ref{lem:Z_QFI_N}, we get:
\begin{corollary}\label{coro:decZ}
    Register-bounded synthesis is decidable for $(\bbZ,<,0)$.
\end{corollary}

Then, since $(\bbZ,<,0)$ reduces to $(\bbN,<,0)$, and reducibility is transitive, we get, by Lemma~\ref{lem:Zk_QFI_Z} and Theorems~\ref{thm:QFI_red} and~\ref{thm:reducibility_synthesis}:
\begin{corollary}\label{coro:decZk}
    Register-bounded synthesis is decidable for $(\bbZ^d,=^d,<^d,0^d)$.
\end{corollary}
\begin{remark*}
  One can similarly show that $\bbN^d$ reduces to $\bbN$. More generally, the above method allows one to lift decidability of register-bounded synthesis to tuples of data values where predicates are applied component-wise. Besides, note that $\bbN^d$ also reduces to $\bbZ^d$, by restricting $\bbZ^d$ to nonnegative values.
\end{remark*}

\subsection{Finite Strings with the Prefix Relation}

In this section, we show that synthesis is decidable over the data domain
$(\Sigma^*, =, \prec, \epsilon)$, where $\Sigma$ is a finite alphabet and
$\prec$ denotes the prefix relation, leveraging a result
of~\cite{DD16} that encodes prefix constraints as integer ones. This still requires some work, as we cannot use the notion of interpretation: a string valuation is encoded as an integer valuation with a \emph{quadratic} number of registers. In the
sequel, $\Sigma$ is a fixed finite set of size $l \geq 2$.

First, $(\Sigma^*, =, \prec, \epsilon)$ reduces to the richer
domain $(\Sigma^*, =, \clen_=, \clen_<, \epsilon)$, where, given $u,v \in \Sigma^*$,
$\clen(u,v)$ denotes the length of the longest common prefix of $u$ and $v$,
and, for $\tr \in \{<,=\}$,
$\clen_{\tr}(u,v,u',v')$ holds whenever $\clen(u,v) \tr \clen(u',v')$. The
reduction is direct, and follows the same lines as \cite[Lemma~3]{DD16}: $u
\prec v$ is encoded as $(\clen(u,u) = \clen(u,v)) \land (\clen(u,u) < \clen(v,v))$, and $K$ is a morphism
on tests and the identity over assignments.
\begin{lemma}
  $(\Sigma^*, =, \prec, \epsilon)$ reduces to $(\Sigma^*, =, \clen_=, \clen_<, \epsilon)$.
\end{lemma}
Note also that satisfiability of tests over both domains is decidable, and NP-complete~\cite[Lemma~7]{DD16}.
%
It now remains to show that $(\Sigma^*, =, \clen_=,\clen_<, \epsilon)$ reduces to
$(\bbN,=,<,0)$. The proof draws on ideas similar to that
of~\cite[Lemmas~8,9]{DD16}, which mainly relies on~\cite[Lemmas~5,6]{DD16}.
Here, it remains to lift them to our synthesis framework, and ensure that
feasibility is preserved despite the dependencies induced by registers.
\begin{lemma}
  \label{lem:prefix_red_N}
  $(\Sigma^*,=, \clen_=, \clen_<, \epsilon)$ reduces to $(\bbN,=, <, 0)$.
\end{lemma}
\begin{proof}
  We describe the main ideas of the proof; a full proof can be found in
  Appendix~\ref{app:lem:prefix_red_N}. %
  From~\cite[Lemma~5,6]{DD16}, we know that a string valuation is characterised
  by the length of the longest common prefixes of all its pairs of values, when
  prefix constraints are concerned. This allows to encode $\Sigma^*$ in $\bbN$:
  given a set $R$ of registers, we introduce a register $\clenR{r}{s}$ for each
  $(r,s) \in R' = (R \cup \{\inWord\})^2$, where $\inWord$ is an additional
  register name that denotes the input data value $\star$ in $\Sigma^*$. Along
  the execution, a register 
  $\clenR{r}{s}$ is meant to contain $\clen(\v(r),\v(s))$. Note that in
  particular, $\clenR{r}{r}$ contains the length of the word stored in $r$. At
  each step, we read a sequence of $\size{R}$ integers that each corresponds to
  the value of $\clen(\star,r)$ for some $r \in R$, that we store in the corresponding
  register $\clenR{\star}{r}$. We then check that they satisfy the $\clen$
  constraints, as 
  well as the properties of~\cite[Proposition 2]{DD16}. The latter consist
  in logical formulas that can be encoded as tests in $(\bbN,=,<,0)$, as they
  only use $=$ and $<$.

  Using~\cite[Lemma~6]{DD16}, from a sequence of integer valuations (called
  \emph{counter valuations} in~\cite{DD16}) that satisfy those properties, we
  can reconstruct a sequence of string valuations. As the integer valuations
  additionally satisfy the $\clen$ constraints, so does the string
  valuations. Thus, if an image $R'$-action word is
  feasible, the original action word is feasible. The converse direction is
  easier: given a sequence $\v_0 \v_1 \dots$ of string valuations that is
  compatible with the $R$-action word, at step $i$ one fills each $\clenR{r}{s}$
  with $\clen(\v_i(r),\v_i(s))$.
\end{proof}
By Theorems~\ref{thm:dec_synt_N} and~\ref{thm:reducibility_synthesis}, we get:
\begin{corollary}\label{coro:decpref}
    Register-bounded synthesis is decidable for $(\Sigma^*,=,\prec,\epsilon)$.
\end{corollary}

\begin{remark}[Complexity analysis]\label{rem:complexity}
Note that the data domains in
Corollaries~\ref{coro:decZ},~\ref{coro:decZk} and~\ref{coro:decpref}
all reduce to $(\bbN,<, 0)$ (all via some rational relations $K$
depending on a set of registers $R$). The time complexities of those
corollaries depend on the complexities of constructing, given a set of
registers $R$, a nondeterministic B\"uchi automaton recognising
$K^{-1}(\QF_R^{(\bbN,<, 0)})$ for all the rational relations $K$
defined in the proofs of those corollaries. It
can be seen from those proofs that for any such rational relation $K$,
it is possible to construct a nondeterministic B\"uchi transducer $A_K$
with polynomially many states in $|R|$ recognising $K$. By taking the synchronized
product of $A_K$ with a nondeterministic automaton recognising
$\QF_R^{(\bbN,<, 0)}$, say of size $n_{qf}$, and by projecting it on
its inputs, one obtains a nondeterministic B\"uchi automaton
recognising $K^{-1}(\QF_R^{(\bbN,<, 0)})$. It can be computed in time
$poly(n_{qf})$. By Fact~\ref{lem:N-goodness-atm} and
Theorem~\ref{thm:gendec}, one gets that the time complexities of
$k$-register-bounded synthesis for data domains $(\mathbb{Z},=,<,0)$,
$(\bbZ^d,=^d,<^d,0^d)$ (for a fixed $d$) and
$(\Sigma^*,=,\prec,\epsilon)$ is doubly exponential in $k$ and $r$ the
number of registers of the specification, and singly exponential in
the number of states of the URA and its number of priorities.
\end{remark}

\section{Conclusion}
We have shown that register-bounded synthesis from
specifications expressed by universal register-automata over $(\bbN, <, 0)$ is decidable within the same time complexity class as the case of URA
over $(\bbN, =)$, completing the picture on synthesis from
register automata over $(\bbN, =)$ and $(\bbN, <,0)$: (unbounded) synthesis is undecidable
for nondeterministic register automata~\cite{EFR21}, decidable for
deterministic register automata over $(\bbN, =)$~\cite{EFR21} and over
$(\bbN, <)$~\cite{EFK21}, and register-bounded synthesis is decidable for URA over
$(\bbN, =)$~\cite{KMB18,EFR21,KK19} and $(\bbN, <,0)$ (this paper), and undecidable for
nondeterministic register automata~\cite{EFR21}.
We also get decidability for the data domains of integers, of tuples
of integers and of finite words with the prefix relation, by reducing
them to $(\bbN,<,0)$. A simple complexity analysis
(Remark~\ref{rem:complexity}) yields a doubly exponential decision
procedure for register-bounded synthesis over these
domains. Systematising this complexity analysis calls for a notion of
polynomial reduction between data domains, that we leave for future
work.

There are other challenging future research directions:
first, universal automata, as argued in the introduction, are well
suited for synthesis, and have been show in the register-free setting
to be amenable to synthesis procedures which are feasible in
practice~\cite{KV05c,SF07,FJR09,DBLP:reference/mc/BloemCJ18}.
We plan
on investigating extensions of these works to the register setting. In
particular, our synthesis algorithm first reduces the problem to a
synthesis problem over a \emph{finite} alphabet with a specification
given by a universal co-B\"uchi automaton. The latter problem is
classically solved by reduction to a parity game obtained by
determinising the universal co-B\"uchi automaton, e.g.\ by using
Safra's determinization procedure. It is an interesting question whether Safraless
procedures from~\cite{KV05c,SF07,FJR09} could be combined with our
game reduction to get more practical algorithms. Another challenging research direction is to consider
synthesis problems from logical specifications instead of
automata, as the nice correspondences between automata and logics for
word languages over finite alphabets do not carry over to data
words. Nevertheless, URA encompass Constraint LTL~\cite{ST11}, and we believe their expressive power could allow one to target other temporal-like logics with data.

\bibliography{refs}

\begin{thebibliography}{10}

\bibitem{AK20}
Shaull Almagor and Orna Kupferman.
\newblock Good-enough synthesis.
\newblock In Shuvendu~K. Lahiri and Chao Wang, editors, {\em Computer Aided
  Verification - 32nd International Conference, {CAV} 2020, Los Angeles, CA,
  USA, July 21-24, 2020, Proceedings, Part {II}}, volume 12225 of {\em Lecture
  Notes in Computer Science}, pages 541--563. Springer, 2020.
\newblock \href {https://doi.org/10.1007/978-3-030-53291-8\_28}
  {\path{doi:10.1007/978-3-030-53291-8\_28}}.

\bibitem{DBLP:conf/fossacs/BerardBLS20}
B{\'{e}}atrice B{\'{e}}rard, Benedikt Bollig, Mathieu Lehaut, and Nathalie
  Sznajder.
\newblock Parameterized synthesis for fragments of first-order logic over data
  words.
\newblock In {\em {FOSSACS}}, volume 12077 of {\em Lecture Notes in Computer
  Science}, pages 97--118. Springer, 2020.

\bibitem{DBLP:books/lib/Berstel79}
Jean Berstel.
\newblock {\em Transductions and context-free languages}, volume~38 of {\em
  Teubner Studienb{\"{u}}cher : Informatik}.
\newblock Teubner, 1979.
\newblock URL: \url{https://www.worldcat.org/oclc/06364613}.

\bibitem{DBLP:reference/mc/BloemCJ18}
Roderick Bloem, Krishnendu Chatterjee, and Barbara Jobstmann.
\newblock Graph games and reactive synthesis.
\newblock In {\em Handbook of Model Checking}, pages 921--962. Springer, 2018.

\bibitem{BC06}
M.~Boja\'nczyk and T.~Colcombet.
\newblock Bounds in $\omega$-regularity.
\newblock In {\em Proc.\ 21st IEEE Symp. on Logic in Computer Science}, pages
  285--296, 2006.

\bibitem{Bojanczyk19}
Miko{\l}aj Boja{\'n}czyk.
\newblock {\em Slightly infinite sets}.
\newblock Miko{\l}aj Boja{\'n}czyk, 2019.
\newblock URL: \url{https://www.mimuw.edu.pl/~bojan/paper/atom-book}.

\bibitem{BL69}
J.R. B{\"u}chi and L.H. Landweber.
\newblock Solving sequential conditions by finite-state strategies.
\newblock {\em Trans. AMS}, 138:295--311, 1969.

\bibitem{BKY00}
Alex Bystrov, David~John Kinniment, and Alexandre Yakovlev.
\newblock Priority arbiters.
\newblock In {\em Proceedings Sixth International Symposium on Advanced
  Research in Asynchronous Circuits and Systems (ASYNC 2000)(Cat. No.
  PR00586)}, pages 128--137. IEEE, 2000.

\bibitem{CJKLS17}
C.S. Calude, S.~Jain, B.~Khoussainov, W.~Li, and F.~Stephan.
\newblock Deciding parity games in quasipolynomial time.
\newblock In {\em Proc.\ 49th ACM Symp. on Theory of Computing}, pages
  252--263, 2017.

\bibitem{DD16}
St{\'e}phane Demri and Morgan Deters.
\newblock Temporal logics on strings with prefix relation.
\newblock {\em Journal of Logic and Computation}, 26(3):989--1017, 2016.

\bibitem{DD07}
St\'ephane Demri and Deepak D'Souza.
\newblock An automata-theoretic approach to constraint ltl.
\newblock {\em Information and Computation}, 205(3):380--415, 2007.
\newblock URL:
  \url{https://www.sciencedirect.com/science/article/pii/S0890540106001076},
  \href {https://doi.org/https://doi.org/10.1016/j.ic.2006.09.006}
  {\path{doi:https://doi.org/10.1016/j.ic.2006.09.006}}.

\bibitem{ESK14}
R.~Ehlers, S.~Seshia, and H.~Kress-Gazit.
\newblock Synthesis with identifiers.
\newblock In {\em Proc. 15th Int. Conf. on Verification, Model Checking, and
  Abstract Interpretation}, volume 8318 of {\em Lecture Notes in Computer
  Science}, pages 415--433. Springer, 2014.

\bibitem{ExibardThesis}
L{\'e}o Exibard.
\newblock {\em {Automatic Synthesis of Systems with Data}}.
\newblock {PhD Thesis}, {Aix-Marseille Universit{\'e} (AMU); Universit{\'e}
  libre de Bruxelles (ULB)}, September 2021.
\newblock URL: \url{http://www.icetcs.ru.is/leoe/files/Exibard_ASSD_SASD.pdf}.

\bibitem{EFK21}
L\'{e}o Exibard, Emmanuel Filiot, and Ayrat Khalimov.
\newblock {Church Synthesis on Register Automata over Linearly Ordered Data
  Domains}.
\newblock In Markus Bl\"{a}ser and Benjamin Monmege, editors, {\em STACS 2021},
  volume 187 of {\em Leibniz International Proceedings in Informatics
  (LIPIcs)}, pages 28:1--28:16, Dagstuhl, Germany, 2021. Schloss Dagstuhl --
  Leibniz-Zentrum f{\"u}r Informatik.
\newblock URL: \url{https://drops.dagstuhl.de/opus/volltexte/2021/13673}, \href
  {https://doi.org/10.4230/LIPIcs.STACS.2021.28}
  {\path{doi:10.4230/LIPIcs.STACS.2021.28}}.

\bibitem{EFR21}
L\'eo Exibard, Emmanuel Filiot, and Pierre-Alain Reynier.
\newblock {Synthesis of Data Word Transducers}.
\newblock {\em {Logical Methods in Computer Science}}, {Volume 17, Issue 1},
  March 2021.
\newblock URL: \url{https://lmcs.episciences.org/7279}, \href
  {https://doi.org/10.23638/LMCS-17(1:22)2021}
  {\path{doi:10.23638/LMCS-17(1:22)2021}}.

\bibitem{FK20}
Rachel Faran and Orna Kupferman.
\newblock On synthesis of specifications with arithmetic.
\newblock In Alexander Chatzigeorgiou, Riccardo Dondi, Herodotos Herodotou,
  Christos Kapoutsis, Yannis Manolopoulos, George~A. Papadopoulos, and Florian
  Sikora, editors, {\em SOFSEM 2020: Theory and Practice of Computer Science},
  pages 161--173, Cham, 2020. Springer International Publishing.

\bibitem{FJR09}
E.~Filiot, N.~Jin, and J.-F. Raskin.
\newblock An antichain algorithm for {LTL} realizability.
\newblock In {\em Proc. 21st Int. Conf. on Computer Aided Verification}, volume
  5643, pages 263--277, 2009.

\bibitem{DBLP:conf/icalp/FiliotJLW16}
Emmanuel Filiot, Isma{\"{e}}l Jecker, Christof L{\"{o}}ding, and Sarah Winter.
\newblock On equivalence and uniformisation problems for finite transducers.
\newblock In Ioannis Chatzigiannakis, Michael Mitzenmacher, Yuval Rabani, and
  Davide Sangiorgi, editors, {\em 43rd International Colloquium on Automata,
  Languages, and Programming, {ICALP} 2016, July 11-15, 2016, Rome, Italy},
  volume~55 of {\em LIPIcs}, pages 125:1--125:14. Schloss Dagstuhl -
  Leibniz-Zentrum f{\"{u}}r Informatik, 2016.
\newblock \href {https://doi.org/10.4230/LIPIcs.ICALP.2016.125}
  {\path{doi:10.4230/LIPIcs.ICALP.2016.125}}.

\bibitem{FKPS19}
B.~Finkbeiner, F.~Klein, R.~Piskac, and M.~Santolucito.
\newblock Temporal stream logic: Synthesis beyond the bools.
\newblock In {\em Proc. 31st Int. Conf. on Computer Aided Verification}, 2019.

\bibitem{GKS10}
O.~Grumberg, O.~Kupferman, and S.~Sheinvald.
\newblock Variable automata over infinite alphabets.
\newblock In {\em Proc. 4th Int. Conf. on Language and Automata Theory and
  Applications}, volume 6031 of {\em Lecture Notes in Computer Science}, pages
  561--572. Springer, 2010.

\bibitem{JPR18}
Ranjit Jhala, Andreas Podelski, and Andrey Rybalchenko.
\newblock Predicate abstraction for program verification.
\newblock In Edmund~M. Clarke, Thomas~A. Henzinger, Helmut Veith, and Roderick
  Bloem, editors, {\em Handbook of Model Checking}, pages 447--491. Springer,
  2018.
\newblock \href {https://doi.org/10.1007/978-3-319-10575-8\_15}
  {\path{doi:10.1007/978-3-319-10575-8\_15}}.

\bibitem{KF94}
M.~Kaminski and N.~Francez.
\newblock Finite-memory automata.
\newblock {\em Theoretical Computer Science}, 134(2):329--363, 1994.

\bibitem{KZ08}
M.~Kaminski and D.~Zeitlin.
\newblock Extending finite-memory automata with non-deterministic reassignment.
\newblock In {\em AFL}, pages 195--207, 2008.

\bibitem{KMB18}
A.~Khalimov, B.~Maderbacher, and R.~Bloem.
\newblock Bounded synthesis of register transducers.
\newblock In {\em 16th Int. Symp. on Automated Technology for Verification and
  Analysis}, volume 11138 of {\em Lecture Notes in Computer Science}, pages
  494--510. Springer, 2018.

\bibitem{KK19}
Ayrat Khalimov and Orna Kupferman.
\newblock Register-bounded synthesis.
\newblock In Wan Fokkink and Rob van Glabbeek, editors, {\em 30th International
  Conference on Concurrency Theory, {CONCUR} 2019, August 27-30, 2019,
  Amsterdam, the Netherlands}, volume 140 of {\em LIPIcs}, pages 25:1--25:16.
  Schloss Dagstuhl - Leibniz-Zentrum f{\"{u}}r Informatik, 2019.
\newblock URL: \url{https://doi.org/10.4230/LIPIcs.CONCUR.2019.25}.

\bibitem{KV05c}
O.~Kupferman and M.Y. Vardi.
\newblock Safraless decision procedures.
\newblock In {\em Proc.\ 46th IEEE Symp. on Foundations of Computer Science},
  pages 531--540, 2005.

\bibitem{MB21}
Benedikt Maderbacher and Roderick Bloem.
\newblock Reactive synthesis modulo theories using abstraction refinement,
  2021.
\newblock \href {http://arxiv.org/abs/2108.00090} {\path{arXiv:2108.00090}}.

\bibitem{NSV01}
F.~Neven, T.~Schwentick, and V.~Vianu.
\newblock Towards regular languages over infinite alphabets.
\newblock In {\em 26th Int. Symp. on Mathematical Foundations of Computer
  Science}, pages 560--572. Springer-Verlag, 2001.

\bibitem{Pit06}
N.~Piterman.
\newblock From nondeterministic {B}{\"{u}}chi and {S}treett automata to
  deterministic parity automata.
\newblock In {\em Proc.\ 21st IEEE Symp. on Logic in Computer Science}, pages
  255--264. IEEE press, 2006.

\bibitem{PR89a}
A.~Pnueli and R.~Rosner.
\newblock On the synthesis of a reactive module.
\newblock In {\em Proc.\ 16th ACM Symp. on Principles of Programming
  Languages}, pages 179--190, 1989.

\bibitem{SF07}
S.~Schewe and B.~Finkbeiner.
\newblock Bounded synthesis.
\newblock In {\em 5th Int. Symp. on Automated Technology for Verification and
  Analysis}, volume 4762 of {\em Lecture Notes in Computer Science}, pages
  474--488. Springer, 2007.

\bibitem{ST11}
Luc Segoufin and Szymon Torunczyk.
\newblock Automata-based verification over linearly ordered data domains.
\newblock In {\em 28th International Symposium on Theoretical Aspects of
  Computer Science (STACS 2011)}. Schloss Dagstuhl-Leibniz-Zentrum fuer
  Informatik, 2011.

\bibitem{Tze11}
N.~Tzevelekos.
\newblock Fresh-register automata.
\newblock In {\em Proc.\ 38th ACM Symp. on Principles of Programming
  Languages}, pages 295--306, New York, NY, USA, 2011. ACM.

\end{thebibliography}
\newpage

\appendix

\section{Detailed Proofs of Section~\ref{sec:dec-recipe}}

\subsection{Proof of Lemma~\ref{lem:Wqf-is-regular}}\label{app:lem:Wqf-is-regular}

\begin{proof}
  \newcommand\dualSs{\widetilde{S}_\mathit{synt}}
  \newcommand\AQF{\mathit{QF}}
  We prove the second claim, implying the first one.
  Let $R = R_S\uplus R_k$.
  Let $\AQF$ be a nondet.\ B\"uchi automaton recognising $\QF_R$ with the sizes mentioned by the lemma.
  Note that $|\Tst_R|$ is $exp(r,k)$ and $|\Asgn_R| = 2^{r+k}$ is $exp(r,k)$.
  Recall that the specification automaton $S$ works on words
  that interleave input and output data of transducers.
  We make this explicit by assuming the states of $S$ are partitioned into input and output states
  for reading input and output data.
  The automaton $S$ can be transformed into this form at the cost of doubling the number of states and in time $poly(n,exp(r))$.
  The transitions from input states are called input transitions,
  from output states -- on output transitions.
  The input alphabet is that of the possible letters on input transitions,
  output alphabet -- on output transitions.
  We now construct the sought universal automaton for $\Wqf$.

  \textbf{1.}
  Construct the nondet.\ B\"uchi automaton $\AQF'$ from $\AQF$ as follows:
  replace every input transition labelled $(\tst,\asgn)\in\Tst_R\x\Asgn_R$
  by transitions labelled $(\tst,\asgn,r^k)$ for every $r^k\in R_k$.
  Hence, the projection of $L(\AQF')$ on $\Tst_R\x\Asgn_R$ equals $\QF_R$.
  Thus, the automaton $\AQF'$ can be constructed in time $poly(N,exp(r,k))$.
  We assume the transition relation of $\AQF'$ is complete
  (to be able to complement it via dualising).
  If it is not complete,
  we can make it such in time $poly(N,exp(r,k))$.

  \textbf{2.}
  Let $\Ss$ be the syntactic view of $S$;
  it is a universal co-B\"uchi automaton with input and output finite alphabets $\Tst_S\x\Asgn_S$.
  We assume that the transition relation of $\Ss$ is complete,
  which can be enforced in time $poly(n,exp(r))$.
  Let $\dualSs$ be the nondet.\ parity automaton dual to $\Ss$,
  thus $L(\dualSs) = \overline{\Ss}$.
  We construct $\dualSs'$ from $\dualSs$ as follows:
  extend the input and output alphabets of $\dualSs$
  from $\Tst_S\x\Asgn_S$ to $\Tst_R\x\Asgn_R$ and to $\Tst_R\x\Asgn_R\x R_k$,
  while preserving the original literals,
  formally as follows.
  Every input transition of $\dualSs$ labelled $(\tstI,\asgnI) \in \Tst_S\x\Asgn_S$
  is replaced by the output transitions labelled $({\tstI}',{\asgnI}') \in \Tst_R\x\Asgn_R$
    s.t.\
    $\tstI\subseteq {\tstI}'$ and $\asgnI\subseteq {\asgnI}'$.
  Similarly,
  every output transition labelled $(\tstO,\asgnO)\in\Tst_S\x\Asgn_S$
  is replaced by the output transitions labelled $({\tstO}',{\asgnO}',r) \in \Tst_R\x\Asgn_R\x R_k$
    s.t.\
    $\tstO\subseteq {\tstO}'$ and $\asgnO\subseteq {\asgnO}'$.
  This can be done in time $poly(n,exp(r,k))$.

  \textbf{3.}
  We now construct $\AQF' \land \dualSs'$.
  First, translate the nondet.\ parity automaton $\dualSs'$ into a nondet.\ B\"uchi automaton with $O(nc)$ states.
  Then, build the product of all automata and get a nondet.\ B\"uchi automaton with $O(Nnc)$ states.
  This can be done in time $poly(N,n,exp(r,k))$.
  Notice that every word accepted by $\AQF' \land \dualSs'$
  is a product of some transducer action word $\ak$ and automaton action word $\aS$ (equipped with labels $r_k$ on output actions to keep track of the output register) such that $\aS$ is rejected by $\Ss$.

  \textbf{4.}
  Project $\AQF' \land \dualSs'$
  into the input alphabet $\Tst_k\x\Asgn_k$ and output alphabet $R_k$.
  This does not affect the number of states and can be done in time $poly(N,n,exp(r,k))$.

  \textbf{5.}
  Shift the component $\Asgn_k$ from the input alphabet to the output alphabet,
  which multiplies the number of states by $|\Asgn_k| = 2^k$,
  so the number of states becomes $O(2^kNnc)$.
  This can be done in time $poly(N,n,exp(r,k))$.
  Call the result $A'$.
  Note that $L(A') = \overline{\Wqf}$.

  \textbf{6.}
  Finally, we treat the nondet.\ B\"uchi automaton $A'$ as universal co-B\"uchi,
  and this is a sought universal co-B\"uchi automaton with $O(2^kNnc)$ many states.

  \smallskip\noindent
  The correctness should be clear as the construction follows the definition of $\Wqf$.
\end{proof}

\section{Detailed Proofs of Section~\ref{sec:applications}}
\label{app:applications}

\subsection{Proof of Lemma~\ref{lem:lab_red_nolab}}
\begin{proof}
    It remains to show that the construction is correct. First, $K$ is rational: a finite
    transducer just has to first output $\overline{b_\Sigma}$ (which
    is independent from $\a$) and then, whenever it reads 
    $(\tst_i,\asgn_i)$ in $\a$, it outputs any
    $(\tst^{lab}_i,\emptyset)(\tst^{data}_i,\asgn_i)$ which satisfies the
    two latter conditions (there are only finitely many). Let us show
    that $K$ preserves feasibility. Suppose $\a$ is feasible
    by  $u = (\sigma_{i_1},\d_1)(\sigma_{i_2},\d_2)\dots$, then since $\D$ is
    assumed to be infinite, there exists an injective mapping $\mu :
    \Sigma\rightarrow \D$ such that $\mu(\sigma_0) = c_0$. Consider the word $v
    = \mu(\sigma_1)\dots 
    \mu(\sigma_n)\mu(\sigma_{i_1})\d_1\mu(\sigma_{i_2})\d_2\dots$. It
    can be checked that there exists a unique action word
    $\overline{b}$ feasible by $v$ of the form
    $\overline{b_\Sigma}.\overline{b_{\a}}$ and by
    construction of $K$ we have $(\a,\overline{b})\in K$. 

    Conversely, if $\overline{b}$ is feasible by some $v =
    e_1e_2\dots e_n\d'_1\d_1\d'_2d_2\dots$ such that
    $(\a,\overline{b})\in K$, then, by letting $e_0 = c_0$ we have, by
    construction of $K$, 
    that $e_i\neq e_j$ for all $0 \leq i < j \leq n$, and for all
    $i\geq 1$, there exists $0\leq j\leq n$ such that $\d'_i =
    e_j$. By letting $\mu : \sigma_i\mapsto e_i$ for all $0\leq i\leq
    n$, we get that $v$ is a $\mu$-encoding of
    $(\mu^{-1}(\d'_1),\d_1)(\mu^{-1}(\d'_2),\d_2)\dots$ which by
    construction of $K$ is a witness of feasibility of $\a$.
\end{proof}

\subsection{Proof of Theorem~\ref{thm:QFI_red}}\label{app:thm:QFI_red}
\begin{proof}
  Let $\D$ and $\D'$ be two data domains such that $\D$ is a quantifier-free interpretation of $\D'$ of dimension $l \geq 1$, with predicates $P$ and constants $C = \{c_0, \dots, c_m\}$, given by formulas $\phi_{\text{domain}}$, $\{\phi_{c} \mid c \in C\}$ and $\{\phi_p \mid p \in P\}$; without loss of generality we assume that $\phi_{c} \Rightarrow \phi_{domain}$ for all $c \in C$. Note that the formulas are not necessarily conjuncts (there can be disjunctions and negations), so a formula might induce multiple tests.

  Let $R$ be a set of registers; we assume without loss of generality that $R$ and $C$ are disjoint and let $d \notin (R \cup C)$ be a fresh register. We define $R' = (R \cup C \cup \{d\}) \times \{1,\dots,l\}$ and write $r^i$ for $(r,i)$. We need to define a feasibility-preserving rational relation $K$ between action words over $R$ and $R'$. The images of an action word $\a$ through $K$ consists in a prefix $\pi$ followed by $\kappa(\a)$, where $\kappa$ is defined at the level of action words and extended homomorphically.

  Prefixes are of the form $\pi = (\top, c_0^1 \cup R \times \{1\}) \dots (\top,c_0^l \cup R \times \{l\}) (\phi'_{c_0}(c_0^1,\dots,c_0^l),\varnothing) \cdot (\top,c_1^1) \dots (\top,c_1^l) (\phi'_{c_1}(c_1^1,\dots,c_1^l), \varnothing) \dots (\top,c_m^1) \dots (\top,c_m^l) (\phi'_{c_m}(c_m^1,\dots,c_m^l), \varnothing)$, where the $\phi'_{c}$ are tests (i.e. conjuncts) that are implied by $\phi_{c}$. It consists in reading the encoding of $c_0$ and storing it in registers $c_0^1, \dots, c_0^l$ as well as in registers of $R \times \{1, \dots, l\}$, and then reading the encoding of the other constants, storing them in the corresponding registers, and checking that the encodings are correct along the run.

  We now define $\kappa$ at the level of actions over $R$. The idea is to first read the $l$ elements of the tuple encoding the input data $\indata$, store it in the $l$ copies of register $d$, check that the read value indeed belongs to the domain and satisfies the test, and store it in the registers of $\asgn$ by reading it a second time. Formally, an action $(\tst,\asgn) \in \Tst_R \times \Asgn_R$ is associated with the sequence of $2l+1$ actions $(\top,d^1) \dots (\top, d^l) \cdot (\phi'_{\text{domain}}(d^1, \dots, d^l) \wedge \tau(\tst)) \cdot (\indata = d^1,\asgn \times \{1\}) \dots (\indata = d^l, \asgn \times \{l\})$, where $\phi'_{\text{domain}}(d^1, \dots, d^l)$ is a test that is a consequence of $\phi_{\text{domain}}(d^1, \dots, d^l)$, and where $\tau$ is defined on atomic tests as $\tau(p(x_1,\dots,x_a)) = \phi_p(x_1^1,\dots,x_1^l, \dots, x_a^1, \dots, x_a^l)$ for all predicates $p \in P$ of arity $a \in \bbN$, replacing $x$ with $d$ if $x = \indata$ and applied pointwise to each conjunct.

  Now, if $\a$ is feasible in $R$ by some sequence $(\d_0^1,\dots,\d_0^l) (\d_1^1,\dots,\d_1^l) \dots$, then $K(\a)$ is feasible by the sequence $c_0^1 \dots c_0^l e_0 \dots c_m^1 c_m^l e_m d_0^1 \dots d_0^l f_0 d_0^1 \dots d_0^l d_1^1 \dots d_1^l f_1 d_1^1 \dots d_1^l \dots$, where $e_0, f_0, e_1, f_1, \dots$ are arbitrary elements of $\bbD$. Conversely, assume that $K(\a)$ is feasible by some sequence $c_0^1 \dots c_0^l e_0 \dots c_m^1 c_m^l e_m d_0^1 \dots d_0^l f_0 {d'}_0^1 \dots {d'}_0^l d_1^1 \dots d_1^l f_1 {d'}_1^1 \dots {d'}_1^l \dots$. Since formulas are checked through the tests, the encoding of each constant $c \in C$ satisfies $\phi_c(c^1, \dots, c^l)$, and after reading $\pi$ registers in $R \times \{1,\dots,l\}$ contain the encoding of $c_0$. Moreover, we necessarily have $d_i^j = {d'}_i^j$. We also know that each $(d_i^1,\dots,d_i^l)$ satisfies $\phi_{\text{domain}}$, and that the tests are satisfied by construction of $\tau$, so it means that $\a$ is feasible by $(\d_0^1,\dots,\d_0^l) (\d_1^1,\dots,\d_1^l) \dots$, $K(\a)$. Moreover, $K$ is rational, as it it defined homomorphically, so $\D$ indeed reduces to $\D'$.
\end{proof}

\subsection{Proof of Lemma~\ref{lem:prefix_red_N}}\label{app:lem:prefix_red_N}
\begin{proof}
  Let $R = \{r_1,\dots,r_k\}$ be a finite set of registers, and let $\a$ be an
  action word over $R$. We let $R'$ be the set of unordered pairs over $R \cup
  \{\inWord\}$, where $\inWord$ is an additional variable name that aims at
  denoting the input word at each step. For readability, a register $\{r,s\} \in
  R'$ is denoted $\clenR{r}{s}$. Note that we only need unordered pairs, since
  $\clen$ is a symmetric function. For clarity, in the following, we import the
  terminology of~\cite{DD16}: a register valuation $\v: R \rightarrow \Sigma^*$ is
  called a \emph{string valuation}, while a register valuation $\v': R'
  \rightarrow \bbN$ is a \emph{counter valuation}. Here, we cannot use the notion of quantifier-free interpretation, since the encoding of a valuation over $R$ necessitates a number of registers that is quadratic in $R$, and interpretations only allow linear encodings.

  Each action $(\tst, \asgn)$ of $\a$
  translates to a sequence of actions as follows:
  \begin{itemize}
  \item First, we read the values of $\clen(\star,r)$ for each $r \in R \cup \{\inWord\}$, and
    store them. This corresponds to a sequence $(\tst_1,\{\clenR{r_1}{\inWord}\})
    \dots
    (\tst_k,\{\clenR{r_k}{\inWord}\})(\tst_{\star},\{\clenR{\inWord}{\inWord}\})$;
    at this point we put no constraints on the incoming data values so
    $\tst_1,\dots,\tst_k,\tst_{\star}$ can be any tests over $(\bbN,<,0)$. Note that the last
    value corresponds to the length of the input data value (which is a word
    over $\Sigma$).
  \item Then, $\a'$ checks that the values that have been read indeed yield a
    string-compatible counter valuation, in the sense of~\cite[Proposition 2 and
    Section 3]{DD16}, i.e. a counter valuation from which one can reconstruct a
    string valuation. Any action $(\tst_{\clen},\varnothing)$
    such that $\tst_{\clen}$ satisfies the following enforces that it is indeed the case:
    \begin{itemize}
    \item If $\clen_=(r,s,r',s') \in \tst$, then $\tst_{\clen}$ should contain $\clenR{r}{s} = \clenR{r'}{s'}$
    \item If $\clen_<(r,s,r',s') \in \tst$, then $\tst_{\clen}$ should contain $\clenR{r}{s} < \clenR{r'}{s'}$
    \item If $r = s \in \tst$, then $\tst_{\clen}$ should contain $\clenR{r}{r} = \clenR{r}{s} = \clenR{s}{s}$
    \item If $r = \epsilon \in \tst$, then $\clenR{r}{r} = 0 \in \tst_{\clen}$
    \item If $\neg \clen_=(r,s,r',s') \in \tst$, then $\tst_{\clen}$ should contain $\clenR{r}{s} < \clenR{r'}{s'}$ or $\clenR{r}{s} > \clenR{r'}{s'}$
    \item If $\neg \clen_<(r,s,r',s') \in \tst$, then $\tst_{\clen}$ should contain $\clenR{r}{s} = \clenR{r'}{s'}$ or $\clenR{r}{s} > \clenR{r'}{s'}$
    \item If $\neg(r = s) \in \tst$, then $\tst_{\clen}$ should contain $\clenR{r}{r} < \clenR{r}{s}$ ($r$ is a strict prefix of $s$) or, symmetrically, $\clenR{s}{s} < \clenR{r}{s}$, or $\clenR{r}{s} < \clenR{r}{r}$ or $\clenR{r}{s} < \clenR{s}{s}$ (they mismatch at some point)
    \item If $\neg(r = \epsilon) \in \tst$, then $\clenR{r}{r} > 0 \in \tst_{\clen}$
    Additionally, we require that $\tst_{\clen}$ implies $\psi_I \wedge
    \psi_{II} \wedge \psi_{III}$, where $\psi_{I},\psi_{II},\psi_{III}$ are
    defined in~\cite[Section 3.2]{DD16} as the syntactical counterparts
    of the formulas $\phi_{I}$, $\phi_{II}$ and $\phi_{III}$
    of~\cite[Proposition~2]{DD16}, and characterise counter valuations that are
    string-compatible~\cite[Lemmas~5,6]{DD16}. We recall them here, so that the
    reader can observe that they only depend on the predicates of
    $(\bbN,<,0)$.
    \begin{itemize}
    \item $\psi_I = \bigwedge_{r,r' \in R \cup \{\inWord\}} (\clenR{r}{r} \geq \clenR{r}{r'})$
    \item $\begin{aligned}
        \psi_{II} = &\bigwedge_{r^0, \dots, r^l \in R \cup \{\inWord\}} \left( \left( \bigwedge_{0 \leq i \leq l} (\clenR{r^0}{r^1} < \clenR{r^i}{r^i}) \right) \wedge \clenR{r^0}{r^1} = \dots = \clenR{r^0}{r^l} \right) \\
        &\Rightarrow \left( \bigwedge_{1 \leq i < j \leq l} (\clenR{r^0}{r^1} < \clenR{r^i}{r^j}) \right)
        \end{aligned}$
    \item $\psi_{III} = \bigwedge_{r,r',r'' \in R \cup \{\inWord\}} (\clenR{r}{r'} < \clenR{r'}{r''}) \Rightarrow (\clenR{r}{r'} = \clenR{r}{r''})$
    \end{itemize}
  \end{itemize}
    Note that all register variables range over $R \cup \{\inWord\}$, as we also
    check the properties for the incoming data value $\star$ (again, recall that they
    consist in words over $\Sigma^*$), as our goal is to be able to reconstruct
    the input sequence of finite strings.
  \end{itemize}
  Note also that we impose
  no condition on $\star$ (the input value in $\bbN$ at this moment): it is
  simply ignored, as we only use this action to conduct the tests.

  This translation of $R$-actions over $(\Sigma^*,<,\epsilon)$ to sequences of
  $R'$-actions over $(\bbN,<,0)$ is extended to action sequences homomorphically.
  Since the resulting relation $K$ is a morphism, it is in particular a rational
  relation. It remains to show that preserves feasibility.

  First, assume that $\a$ is feasible by some data word $x = w_0 w_1 \dots \in
  (\Sigma^*)^\omega$, and let $\v_0 \v_1 \dots$ be a sequence of valuations that
  witnesses this fact, where, for all $i \geq 0$, $\v_i: R \rightarrow
  \Sigma^*$, and $\v_0: r \mapsto \epsilon$. We translate it into the suitable
  sequence of $\clen$ values. Formally, for all $i \geq 0$,
  and for all $0 \leq j \leq k$, we let $n_i^j = \clen(\v_i(r_j),w_i)$. Additionally,
  $n_i^{\star} = \clen(w_i,w_i) = \length{w_i}$, and $n_i^{\clen} \in \bbN$ is
  some value whose choice does not matter (recall that the action corresponding
  to $\tst_{\clen}$ is dedicated to checking that the input values indeed
  satisfy the prefix constraints and the string-compatibility conditions). Let
  us show by induction on $i$ 
  that there exists an action word $\a'$ in $K(\a)$ such that $n_0^0
  \dots n_0^{k} n_0^{\star} n_0^{\clen} n_1^0 \dots n_1^{k} n_1^{\star}
  n_1^{\clen} \dots$ is compatible with $\a'$, and that it yields the
  following sequence of valuations: for all $i \geq 0$
  and all $j \in \{0,\dots,k,\star,\clen\}$, ${\v'}_i^j: \{r,s\} \in {R \choose
    2} \mapsto \clen(\v_i(r),\v_i(s))$ and, for all $0 \leq m \leq k$,
  ${\v'}_i^j(\{r_m,\inWord\}) =
  \left\{\begin{array}{l}\clen(\v_i(r_m),w_i) \text{ if $m \leq j$} \\
           {\v'}_{i-1}^k(r_m,\inWord) \text{ otherwise} \end{array}  \right.$.
       Finally, for all $j \in \{0,\dots,k,\star\}$, ${\v'}_i^j(\{\inWord\}) =
       \length{w_{i-1}}$, and ${\v'}_i^{\clen}(\{\inWord\}) = \length{w_i}$.

       For simplicity, we initialise the induction at $-1$ by picking $w_{-1} =
       \epsilon$, which makes the properties true. Now, assume that we have
       built $\a'$ up to step $i_0$. Thus, the current valuation is
       ${\v'}_{i_0}^\clen$, which is such that ${\v'}_{i_0}^\clen: \{r,s\}
       \mapsto \clen(\v_{i_0}(r),\v_{i_0}(s))$. Additionally,
       ${\v'}_{i_0}^\clen(r,\inWord) \mapsto \clen(r,\inWord)$ for all $r \in R
       \cup \{\inWord\}$. Now, the $R'$-action word reads $n_{i_0+1}^0 \dots
       n_{i_0+1}^k n_{i_0+1}^{\star} n_{i_0+1}^{\clen}$, as defined above. Since
       we set no conditions on $\tst_1,\dots,\tst_k,\tst_{\star}$, we can pick
       suitable tests so that $n_{i_0+1}^0 \dots n_{i_0+1}^k n_{i_0+1}^{\star}$
       satisfies them. Finally, $\tst_{\clen}$ does not depend on $\star$, so
       any value of $n_{i_0+1}^{\clen}$ is suitable. It remains to show that at
       this point, ${\v'}_{i_0+1}^{\clen}$ satisfies $\tst_{\clen}$. By
       definition, it is a string-compatible counter valuation, since it is
       obtained from a string valuation. Thus, by \cite[Proposition 2]{DD16}, it
       satisfies properties $\psi_{I}$, $\psi_{II}$ and $\psi_{III}$. Moreover,
       $\v_{i+1}$, along with $w_{i+1}$, satisfies the $\clen$ constraints, so
       ${\v'}_{i_0+1}^{\clen}$  satisfies them as well. Thus, the property holds
       at step $i_0+1$.

  Conversely, assume that some $\a' \in K(\a)$ is feasible by some data word
  $n_0^0 \dots n_0^{k} n_0^{\star} n_0^{\clen} \cdot n_1^0 \dots n_1^{k} n_1^{\star}
  n_1^{\clen} \dots \in \bbN^\omega$, along with a sequence of integer valuations ${\v'}_0^0
  \dots {\v'}_0^{k} {\v'}_0^{\star} {\v'}_0^{\clen} \cdot {\v'}_1^0 \dots
  {\v'}_1^{k} {\v'}_1^{\star} {\v'}_1^{\clen} \dots$. We build by induction on
  $i$ a sequence of string valuations $\v_i: R \rightarrow \Sigma^*$, along with
  a data word $x = w_0 w_1 \dots$ which are compatible with $\a$, and such that
  for all $i \geq 0$, ${\v'}_i^{0}: {R \cup \{\inWord\} \choose 2} \rightarrow
  \bbN$ is a counter valuation that is string-compatible with $\v_i[\star
  \leftarrow w_{i-1}]$, with the convention that $w_{-1} = \epsilon$.

  Initially, $\v_0: R \mapsto \epsilon$ and $w_{-1} = \epsilon$, so ${\v'}_0^0: \{r,s\} \mapsto 0$ for all $r,s \in R \cup \{\inWord\}$, is string-compatible with $\v_0$.

  Now, assume that we have built the $\v_j$ up to step $i \geq 0$. By
  construction of $K(\a')$, we know that ${\v'}_{i}^{l+1}$ satisfies $\psi_I
  \wedge \psi_{II} \wedge \psi_{III}$, as we ask that $\tst_{\clen}$ implies
  them, and is such that for all $r,s \in R$, ${\v'}_{i}^{\clen}(\{r,s\}) =
  \clen(\v_i(r),\v_i(s))$. Indeed, ${{\v'}_{i}^{\clen}}_{\mid R} =
  {{\v'}_i^0}_{\mid R}$, as the values of the $\clenR{r}{s}$ for $r,s \in R$ are
  left untouched before updating the registers when transitioning to
  ${\v'}_{i+1}^0$). In other words,
  the counters in ${\v'}_i^{\clen}$
  indeed contain the length of the longest common prefixes of the strings of
  $\v_i$ (in the terminology of~\cite{DD16}, this means that ${\v'}_i^{\clen}
  \approx_R \v_i$).
  Then, since ${\v'}_i^{\clen}$ additionally satisfies conditions
  $\psi_{I}$, $\psi_{II}$ and $\psi_{III}$ with regards to $\inWord$,
  by~\cite[Lemma 6]{DD16}, we can construct a value for $\inWord$ that is
  consistent with $\v_i$.
  More precsiely, by applying \cite[Lemma 6]{DD16} to $X =
  \{\inWord\}$, we know that there exists a string $w_i$ such that for all
  ${\v'}_{i}^{\clen} \approx_{R \cup \{\inWord\}} \v_i[\star \leftarrow w_i]$,
  i.e. for all $r,s \in R \cup \{\inWord\}$ (note the addition of $\inWord$), we
  have that ${\v'}_{i}^{\clen}(\{r,s\}) = \clen(\v_i(r),\v_i(s))$, where
  $\v_i(\star) = w_i$. Moreover, we know that the $\clen$ constraints are
  satisfied, since we encoded them in $\tst_{\clen}$. Thus, $w_0 \dots w_i
  w_{i+1}$ is compatible with $\a$ up to index $i+1$, and the property holds at
  step $i+1$.
\end{proof}

\end{document}